\documentclass[journal,final,twoside,romanappendices]{IEEEtran}
\usepackage[noadjust]{cite}
\usepackage[cmex10]{amsmath}
\usepackage{amssymb,amsfonts,amscd,amsthm,amsbsy}
\usepackage{graphicx}
\usepackage[english]{babel}
\usepackage{mathrsfs}
\usepackage{mathtools}
\usepackage{array}
\usepackage{blkarray}

\usepackage{bbm}

\usepackage{tabulary}
\usepackage{makecell}
\usepackage{xcolor}

\newtheorem{theorem}{Theorem}
\newtheorem{lemma}{Lemma}

\newtheorem{corollary}{Corollary}
\theoremstyle{definition}
\newtheorem{remark}{Remark}

\theoremstyle{definition}

\theoremstyle{definition}
\newtheorem{example}{Example}
\newtheorem{definition}{Definition}
\theoremstyle{definition}

\theoremstyle{definition}

\usepackage{etoolbox}
\AtEndEnvironment{example}{\null\hfill$\square$}%
\AtEndEnvironment{remark}{\null\hfill$\square$}%

\newcommand{\Fb}{\mathbbmss{F}}
\newcommand{\Zb}{\mathbbmss{Z}}
\newcommand{\Cs}{\mathbbmss{M}}
\newcommand{\wt}{\mathsf{wt}}
\newcommand{\rcov}{r_{\rm cov}}
\newcommand{\Xc}{\mathcal{X}}
\newcommand{\rmin}{\mathrm{In}}
\newcommand{\rmout}{\mathrm{Out}}
\newcommand{\ann}{\mathrm{Ann}}
\newcommand{\taunat}{\tau_{\mathrm{nat}}}
\newcommand{\rhoreg}{\rho^{\mathrm{reg}}}
\newcommand{\dft}{\Phi}
\newcommand{\tth}{\text{th}}

\normalsize

\title{Permute \&  Add Network Codes via Group Algebras}
\author{Lakshmi Prasad Natarajan and Smiju Kodamthuruthil Joy%
\thanks{This article was presented in part at the 2021 IEEE International Symposium on Information Theory (ISIT 2021), Melbourne, Australia.}%
\thanks{The authors are with the Department of Electrical Engineering, Indian Institute of Technology Hyderabad, India, email: \{lakshminatarajan,\,ee17resch11017\}@iith.ac.in.}%

\thanks{Copyright (c) 2022 IEEE. Personal use of this material is permitted. However, permission to use this material for any other purposes must be obtained from the IEEE by sending a request to pubs-permissions@ieee.org.}%
}

\begin{document}

\maketitle

\begin{abstract}
A class of network codes have been proposed in the literature where the symbols transmitted on network edges are binary vectors and the coding operation performed in network nodes consists of the application of (possibly several) permutations on each incoming vector and XOR-ing the results to obtain the outgoing vector. These network codes, which we will refer to as \emph{permute-and-add} network codes, involve simpler operations and are known to provide lower complexity solutions than scalar linear network codes. The complexity of these codes is determined by their \emph{degree} which is the number of permutations applied on each incoming vector to compute an outgoing vector. Constructions of permute-and-add network codes for multicast networks are known. In this paper, we provide a new framework based on group algebras to design permute-and-add network codes for arbitrary (not necessarily multicast) networks. Our framework allows the use of any finite group of permutations (including circular shifts, proposed in prior works) and admits a trade-off between coding rate and the degree of the code. Further, our technique permits elegant recovery and generalizations of the key results on permute-and-add network codes known in the literature. 
\end{abstract}


% \emph{The full version of this paper, including the proofs of all the theorems, lemmas and examples, is available online~\cite{}.}

% % % % % 

\section{Introduction} \label{sec:intro}

% We consider network coding over directed acyclic graphs. We assume that the symbols transmitted on the edges are elements of a module, and the local coding coefficients are elements of a ring (over which the module is defined).

% Network coding allows intermediate nodes of a network to perform non-trivial coding operations on the data packets and thereby achieve significantly higher throughput than routing~\cite{ACLY_IT_00}. 
\IEEEPARstart{N}{etwork} coding theory~\cite{ACLY_IT_00} is dominated by the study of linear network codes~\cite{LYC_IT_03,KoM_NET_03,CDFZ_IT_06,JSCEEJT_IT_05,EbF_IT_11}. In \emph{scalar linear network coding} the symbols carried by each network edge is an element of a finite field $\Fb_q$, and is obtained by computing an $\Fb_q$-linear combination of the symbols carried on its parent edges. 
It is well known that scalar linear network coding is sufficient to achieve the capacity of multicast networks as long as the size of the field $\Fb_q$ is sufficiently large~\cite{LYC_IT_03,KoM_NET_03}.
Note that a scalar linear network coding solution requires all the network nodes to perform arithmetic over a (possibly large) finite field.
% Similarly, in vector linear network codes the symbols carried by each edge is a vector over a finite field and the encoding kernels or coefficients are matrices over the same field.

An alternative to scalar linear network coding, which can simplify network coding operations, is to use vector linear network codes where the encoding kernels are linear combinations of permutation matrices~\cite{JCE_ISIT_06,XMA_ISIT07,HaK_ITW_10,HSCL_IT_16,TSLYL_IT_19,STLYL_TCOMM_19,TSYL_COMML_19}. 
For these network codes, the symbols carried by the network edges are length-$n$ binary vectors, and the coding operation performed at a network node is the application of (possibly several) permutations on each incoming binary vector and XOR-ing (which is same as addition in the field $\Fb_2$) the permuted vectors to determine the outgoing binary vector.
Using the vocabulary of~\cite{JCE_ISIT_06} (and by mildly generalizing its terminology), we will refer to such network codes as \emph{permute-and-add} network codes.
The \emph{degree} of a permute-and-add network code is the maximum number of permutations applied on each incoming binary vector to compute an outgoing vector at any node~\cite{STLYL_TCOMM_19,TSLYL_IT_19}. 
Note that the degree determines the number of XORs to be performed at each node. Since the task of performing permutations is cheap, the degree acts as a proxy for the complexity of a permute-and-add network code.
It is known that permute-and-add network codes can provide lower complexity coding operations than scalar linear network coding~\cite{JCE_ISIT_06,TSLYL_IT_19}.

 Permute-and-add network codes are attractive because of the simplicity of the operations required to implement them. A scalar linear solution over a field of size $2^n$ requires the intermediate nodes of the network to perform arithmetic in extensions of the binary field. 
This might involve either polynomial multiplication and modulo operations with respect to a degree $n$ polynomial, or the use of look-up tables for finite field multiplications when the field size is small. 
In contrast, permute-and-add network codes require only bit-wise additions over $\Fb_2$ (which are realized as XOR operations) and circular shifts (note that any permutation can be decomposed as a product of circular shifts). 
These operations can be implemented easily in software and hardware. 
% For example, in circular shift linear network code, a special case of permute-and-add code,  cyclic permutations on binary vectors can be implemented by appropriately modifying a pointer in software or by using shift registers or dedicated barrel shifters in hardware.
As a result, permute-and-add network codes could find applications in networks with computationally limited nodes.

The permute-and-add codes of~\cite{JCE_ISIT_06} were proposed for multicast networks using a random coding framework. These codes employ degree $1$ permute-and-add operation at non-sink nodes, while the decoding matrices are dense, indicating high complexity at sink nodes. 
The prior works~\cite{XMA_ISIT07,HaK_ITW_10,HSCL_IT_16,TSLYL_IT_19,STLYL_TCOMM_19,TSYL_COMML_19} all employ only circular shifts (i.e., cyclic permutations) for coding operations, and following~\cite{TSLYL_IT_19}, we will refer to these permute-and-add network codes as \emph{circular-shift} network codes.
A deterministic circular-shift network code was proposed in~\cite{XMA_ISIT07} for combination networks in which the coding operations performed at non-sink nodes are of degree $1$. The existence of circular-shift network coding solutions for multicast networks was proved in~\cite{HaK_ITW_10}. 
Codes for repairing failed disks in distributed storage systems that make use of circular-shift network codes were proposed in~\cite{HSCL_IT_16}.
Circular-shift network codes were designed in~\cite{TSLYL_IT_19,STLYL_TCOMM_19,TSYL_COMML_19,TSWY_COMML_20} for multicast networks by lifting scalar linear network codes.
A similar code over $\Zb_{256}$ was designed in~\cite{ShH_ISIT20}.
Note that most of the prior works on permute-and-add network codes propose solutions for multicast networks only.
 
The simplicity of permute-and-add network coding solution, in comparison to scalar linear solutions over finite fields, has also been noted in earlier works, please see Section~IV-B of~[11], Section~I of~[15] and Section~IV-A of~[10].

In this paper we provide a new algebraic framework for designing permute-and-add network codes. We use the ring theoretic platform of Connelly and Zeger~\cite{CoZ_Part1_IT_18,CoZ_Part2_IT_18} and show that permute-and-add network codes can be obtained from linear network codes over ideals of group algebras.
Unlike previous works, our technique applies to arbitrary directed acyclic multigraphs (which are not necessarily multicast networks), and both the encoding as well as the decoding procedures of our network codes employ permute-and-add operations.
Further, our framework admits the use of any finite group of permutations (including circular shifts) and allows the designer to trade-off the rate of the network code to achieve a smaller degree. 
The generality of our technique permits us to recover and generalize some of the key results from~\cite{HaK_ITW_10,TSLYL_IT_19,STLYL_TCOMM_19}.

In Section~\ref{sec:network_model} we review the model of linear network coding over rings and present our main observation that underlies the design of permute-and-add network codes of small degree. In Section~\ref{sec:network_coding_over_group_codes} we show that permute-and-add network codes can be obtained as linear network codes over group codes (which are ideals of group algebras). The existence of permute-and-add network coding solutions is dealt with in Section~\ref{sec:existence_of_solutions} where we discuss network codes arising from semi-simple Abelian group algebras (Section~\ref{sub:sec:semi-simple}), circular-shift network codes (Section~\ref{sub:sec:circular-shift}), and network codes designed using non-cyclic permutation groups (Section~\ref{sub:sec:non-cyclic-abelian}).
Our results in Sections~\ref{sec:network_coding_over_group_codes} and~\ref{sec:existence_of_solutions} are on permute-and-add network codes over the binary field. In Section~\ref{sec:arbit_finite_fields} we highlight how our main techniques generalize to codes over arbitrary finite fields.

\emph{Notation:} For integers $a,b$, the symbol $(a,b)$ denotes their gcd. Unless otherwise specified, all vectors are column vectors.

\section{Network Model and The Main Principle} \label{sec:network_model}

We consider linear network coding over a module $M$. 
% i.e., the network coding operations at the non-source nodes and the decoding operations at the sink nodes are linear operations in $M$.
We assume that the ring $R$ underlying the module $M$ contains a multiplicative identity $1$ and is not necessarily commutative. Recall that a \emph{left $R$-module $M$} is an Abelian group $(M,+)$ together with a mapping that sends $R \times M \to M$ such that
\begin{enumerate}
\item[\emph{(i)}] $(r_1 + r_2)m = r_1m + r_2m$
\item[\emph{(ii)}] $r_1(r_2m) = (r_1r_2)m$
\item[\emph{(iii)}] $r(m_1 + m_2) = rm_1 + rm_2$
\item[\emph{(iv)}] $1\,m = m$
\end{enumerate} 
for all $r,r_1,r_2 \in R$ and $m,m_1,m_2 \in M$. 
For brevity, we sometimes refer to $M$ simply as $R$-module (instead of left $R$-module). In order to emphasize the underlying ring, we sometimes denote the module as $\prescript{}{R}{M}$.
We will always assume that the size of the module $|M|$ is finite.

\subsection{Network Model} 

% We consider linear network coding over a module $\prescript{}{R}{M}$, i.e., the network coding operations at non-source nodes and the decoding operations at the sink nodes are linear operations in $\prescript{}{R}{M}$.
We will assume that the network is a directed acyclic multigraph with finitely many nodes and edges. 

We associate with the network code a finite set $\Xc$, called the \emph{edge alphabet}, and an injective map \mbox{$\tau: M \to \Xc$}. Necessarily, \mbox{$|\Xc| \geq |M|$}. We will assume that each edge in the network is a noiseless communication link that carries one symbol from the edge alphabet $\Xc$. The function $\tau$ and its inverse \mbox{$\tau^{-1}: \tau(M) \to M$} are used to represent elements of the module $M$ as symbols from $\Xc$ and vice versa.
The linear coding operations performed at the nodes are over $M$, while the alphabet used for communicating along the network edges is $\Xc$. 

% We will assume that the network is a directed acyclic multigraph with finitely many nodes and edges. 
% Each edge represents a noise-free communication link with unit capacity, i.e., the capacity to carry a symbol from the module $M$. 
A message is an information-bearing random variable taking values in $M$. We assume that there are $s$ messages $Z_1,\dots,Z_s$ generated in the network.
% A network is associated with finitely many messages which are statistically independent of each other. 
Sources are network nodes where the messages are generated, and sinks are the nodes where one or more messages are demanded. 
Each message is demanded by a subset of sink nodes.
The set of incoming edges at a node $v$ will be denoted as $\rmin(v)$ and the set of outgoing edges of $v$ is $\rmout(v)$. 
% Without loss of generality, we assume that a source node generating $s$ messages has $s$ special incoming edges each of which carries the corresponding message symbol. 
% 
Without loss of generality, we assume the following: 
\emph{(i)}~there are exactly $s$ source nodes and there is a one-to-one correspondence between the messages $Z_1,\dots,Z_s$ and the set of source nodes such that each message is generated at the corresponding unique source node;  
\emph{(ii)}~source nodes have no incoming edges; 
and \emph{(iii)}~if $v$ is a source node and $Z_i \in M$ is the message generated at $v$, then every outgoing edge of $v$ carries $\tau(Z_i) \in \Xc$.
% Each message is generated at a unique source node and is demanded by a subset of sink nodes. 
% Note that a node can be both a source and a sink, i.e., it could generate a message while demanding a different message generated at another source.
% 

We will use $X_e \in \Xc$ to denote the symbol carried along the edge $e$.
A coding coefficient \mbox{$k^{d,e} \in R$} is assigned to each pair $(d,e)$ of \emph{adjacent edges}, i.e., if there exists a node $v$ such that \mbox{$d \in \rmin(v)$} and \mbox{$e \in \rmout(v)$}. For every non-source node $v$ and edge $e \in \rmout(v)$, the symbol carried in $e$ is generated by the linear operation
\begin{equation} \label{eq:coding_operation}
\textstyle X_e = \tau \left( \sum_{d \in \rmin(v)} k^{d,e} \, \tau^{-1}(X_d)\right).
\end{equation} 
Similarly, a sink node $v$ demanding a message $Z_i$ uses a linear operation
\begin{equation} \label{eq:decoding_operation}
\textstyle \sum_{d \in \rmin(v)} k^{d,i} \, \tau^{-1}(X_d)
\end{equation} 
to decode $Z_i$, where the coefficients $k^{d,i} \in R$.

\begin{remark}
Except for the notion of the edge alphabet, the network coding model used in this paper is identical to the ring theoretic framework of Connelly and Zeger~\cite{CoZ_Part1_IT_18,CoZ_Part2_IT_18}. That is, when $\Xc=M$ and $\tau$ is the identity map, our model is identical to~\cite{CoZ_Part1_IT_18,CoZ_Part2_IT_18}.
\end{remark}
A \emph{linear network code} over the $R$-module $M$ is the collection of coding coefficients 
$\left\{k^{d,e}~|~(d,e) \text{ are adjacent}\right\}$ and  $\cup_{i}\left\{k^{d,i}~|~\exists v \text{ such that } 
d \in \rmin(v), v \text{ demands } Z_i \right\}$.

We will sometimes refer to such a code as an $\prescript{}{R}{M}$-linear network code.
For brevity, we will represent the collection of encoding and decoding coefficients simply as $\{k^{d,e}\}$ and $\{k^{d,i}\}$, respectively.

A \emph{network coding solution} over the $R$-module ${M}$ (with edge alphabet $\Xc$) is a linear network code where each sink's decoding functions recover its demands. A network code is \emph{solvable} over $\prescript{}{R}{M}$ if it has a linear solution over $\prescript{}{R}{M}$ (for some edge alphabet $\Xc$). 
A network is scalar linearly solvable over $R$ if it is solvable over the module $\prescript{}{R}{R}$.

It is clear that the solvability of a network is independent of the choice of $\Xc$ and $\tau$ as long as the map $\tau$ is injective.
However, the choice of $\Xc$ and $\tau$ determines how the coding operations are realized by the network nodes and determines the communication cost or the rate of the network code (for instance, see Remark~\ref{rem:bit_truncation}).
% As we will show in this paper, the choice of $\Xc$ and $\tau$ can affect the complexity of the encoding and decoding operations at the non-source nodes. Also, the choice of $\Xc$ decides the communication cost or the rate of the network code. 
Since each edge carries one element of $\Xc$ and each message is an element of $M$, the \emph{rate} of the network code is 
\begin{equation*}
 \frac{\log_2 |M|}{\log_2 |\Xc|} = \log_{|\Xc|} |M|.
\end{equation*} 
% Note that 
This definition of rate is in the same spirit as that of fractional network codes where messages are $k$-length vectors and edge symbols are $n$-length vectors over a finite field and the rate is $k/n$~\cite{CDFZ_IT_06}.

\begin{example}
Scalar linear network codes over a finite field $\Fb_q$ have $R=M=\Xc=\Fb_q$, while $k$-dimensional vector linear network codes have $M=\Xc=\Fb_q^k$ and $R=\Fb_q^{k \times k}$ which is the ring of $k \times k$ matrices over $\Fb_q$. In both cases $M=\Xc$ and $\tau$ is the identity map.
\end{example}

\begin{example} \label{example:Shum_Hou}
Shum and Hou~\cite{ShH_ISIT20} designed an array code for distributed storage across $6$ disks such that the original data can be recovered from any $4$ disks. In the context of network coding, this is a code for the $\binom{6}{4}$ combination network~\cite{NgY_ITW04}. 
The coded data stored in the $6$ disks of the array code correspond to the network coded packets transmitted from the source to the six nodes in the first layer of the combination network.
For this code, $R$ is the commutative ring \mbox{$\Zb_{256}[x]/(x^5-1)$} and $M$ is the ideal of $R$ consisting of all \mbox{$a(x) \in \Zb_{256}[x]/(x^5-1)$} such that \mbox{$a(x) \mod (x-1)=0$}. Note that any ideal of a ring is a module over the same ring, and hence, $M$ is an $R$-module. In Shum and Hou's code, \mbox{$\Xc=\Zb_{256}^{4}$} and $\tau$ maps \mbox{$a(x) = a_0 + a_1x + a_2x^2 + a_3x^3 + a_4x^4 \in M$} to \mbox{$(a_0,a_1,a_2,a_3) \in \Xc$}. This is an injective map since \mbox{$a_4 = -(a_0 + \cdots + a_3)$}, over $\Zb_{256}$, for every $a(x) \in M$. In this case, $|M|=|\Xc|=256^4$ and the rate is $1$. 
\end{example}

\subsection{The Main Principle}

We now present a simple observation that plays a key role in our framework for designing permute-and-add network codes with small degree.

The \emph{annihilator} of a left $R$-module $M$ is the set of elements of $R$ whose action on $M$ is identical to the action of the zero element of $R$ on $M$, 
\begin{equation*}
\ann(M) = \left\{ r \in R~|~rm = 0 \text{ for all } m \in M \right\}.
\end{equation*} 
The module $\prescript{}{R}{M}$ is \emph{faithful} if $\ann(M)=\{0\}$ and \emph{unfaithful} otherwise.
We now observe that if the elements of $\ann(M)$ are used to perturb the encoding and decoding coefficients of a network coding solution, then the resulting network code is also a solution.

Let $\{k^{d,e}\}$ and $\{k^{d,i}\}$ be the encoding and decoding coefficients of a given linear network code over a module $M$. 
For each pair of adjacent edges $(d,e)$, we choose a new encoding coefficient $k^{d,e}+a^{d,e}$ where, the original coefficient $k^{d,e}\in R$ is modified by the addition of an arbitrary element $a^{d,e} \in \ann(M)$. 
We know that $(k^{d,e}+a^{d,e})\,\tau^{-1}(X_d)=k^{d,e}\,\tau^{-1}(X_d)$ since $\tau^{-1}(X_d) \in M$ and $a^{d,e} \in \ann(M)$. 
Similarly for each node $v$ demanding a message $Z_i$ and each $d \in \rmin(v)$, 
we can change the decoding coefficient $k^{d,i}$ to $k^{d,i} + a^{d,i}$, where $a^{d,i} \in \ann(M)$, without affecting the decodability at $v$.
% $k^{d,i}+a^{d,i}$ is the new decoding coefficient, where $k^{d,i}\in R$   and $a^{d,i} \in \ann(M)$  and hence, $(k^{d,i} + a^{d,i})\, \tau^{-1}(X_d)= k^{d,i}\, \tau^{-1}(X_d)$.
% Thus, we  modify the original network code by using the coefficients $\{k^{d,e} + a^{d,e}\}$ and $\{k^{d,i} + a^{d,i}\}$ instead of $\{k^{d,e}\}$ and $\{k^{d,i}\}$, respectively.
Thus, any network coding solution $\{k^{d,e}\},\{k^{d,i}\}$ can be modified to another solution $\{k^{d,e} + a^{d,e}\}$, $\{k^{d,i} + a^{d,i}\}$.
We state this observation as Lemma~\ref{lem:annihilator}. Its formal proof is included in Appendix~\ref{app:lem:annihilator}.

\begin{lemma} \label{lem:annihilator}
Consider an $R$-linear solution to a network over a module $M$. 
The network code obtained by adding arbitrary elements of $\ann(M)$ to the coding coefficients of this network code is also a solution to this network.
% For any choice of $a^{d,e}, a^{d,i} \in \ann(M)$, if the original network code $\{k^{d,e}\}$, $\{k^{d,i}\}$ is a solution to the given network, then the modified network code $\{k^{d,e} + a^{d,e}\}$, $\{k^{d,i} + a^{d,i}\}$ is also a solution to this network.
\end{lemma}

Lemma~\ref{lem:annihilator} implies that if $\ann(M) \neq \{0\}$ (in other words, if $M$ is an unfaithful $R$-module) then there is flexibility in the choice of encoding and decoding coefficients. This choice can be exercised in such a way that the resulting network coding operations are of low complexity. We proceed by applying this technique to \emph{group codes}~\cite{Ber_Cybernetics_67,KeS_ContemporaryMath_01} in Section~\ref{sec:network_coding_over_group_codes}.

\section{Network Coding over Group Codes} \label{sec:network_coding_over_group_codes}

We briefly review group codes, group algebras and their matrix representation in Section~\ref{sub:sec:group-algebras}, and use these tools to obtain permute-and-add network codes in Section~\ref{sub:sec:permute-and-add}.

\subsection{Review of Group Algebras and their Matrix Representation} \label{sub:sec:group-algebras}

% We first recall the definition of a group algebra.
Let $G$ be a finite group (not necessarily commutative) and $\Fb_2=\{0,1\}$ the finite field of size $2$. The group algebra $\Fb_2[G]$ is the set of all possible formal sums $\sum_{g \in G} a_g g$, where $a_g \in \Fb_2$. Addition and multiplication in $\Fb_2[G]$ are defined as 
\begin{align*} 
\sum_{g \in G}a_g g + \sum_{g \in G} b_g g = \sum_{g \in G} \left(a_g + b_g\right)g, \text{ and }
\end{align*} 
\begin{align}
 \label{eq:addition_multiplication_group_algebra}   
 \left(\sum_{g \in G}a_g g\right) \cdot \left(\sum_{g \in G} b_g g\right) = \sum_{g \in G} \left(\sum_{h \in G} a_h b_{h^{-1}\!g}\right) g.
\end{align}
The ring $\Fb_2[G]$ is commutative if and only if $G$ is Abelian.  

A \emph{group code} $\Cs$ is a left-ideal of $\Fb_2[G]$, i.e., $\Cs$ is a subgroup of $(\Fb_2[G],+)$ such that $rm \in \Cs$ for all $r \in \Fb_2[G]$ and $m \in \Cs$. Any group code $\Cs$ is a left $\Fb_2[G]$-module where the action of a ring element on $\Cs$ is the same as the product of elements in $\Fb_2[G]$.

% \subsubsection{Matrix Representation of the Group Algebra}

If \mbox{$|G|=n$}, there is a natural $\Fb_2$-linear embedding $\taunat: \Fb_2[G] \to \Fb_2^n$ that maps \mbox{$m = \sum_{g \in G}m_g g$} to the column vector $(m_g)_{g \in G}$ using some fixed ordering of elements of $G$, 
\begin{equation*}
\taunat\left( \sum_{g \in G} m_g g\right) = \left(m_g\right)_{g \in G}.
\end{equation*} 
Clearly, the set of vectors $\{\taunat(1 g)~|~g \in G\}$ form the standard basis for $\Fb_2^n$, where $\taunat(1 g)$ is an $n$-length column vector that contains a $1$ in the $g^{\text{th}}$ position and zeros elsewhere.
% For our network code we will use \mbox{$\Xc=\Fb_2^n$} as the edge alphabet and the restriction of $\taunat$ to $\Cs$ as the associated injective map. 

The \emph{regular representation}~\cite{Ser_Springer_77} of $G$ in $\Fb_2^n$ maps each $g \in G$ to a permutation matrix $\rhoreg_g \in \Fb_2^{n \times n}$. If the rows and columns of $\rhoreg_g$ are indexed by the elements of $G$, the entry in the $k^{\text{th}}$ row and $h^{\text{th}}$ column of $\rhoreg_g$ is 
\begin{align} \label{eq:reg_group_rep}
\rhoreg_g(k,h) = 1 \text{ if } k = gh, \text{ and } \rhoreg_g(k,h) = 0 \text{ otherwise}.
\end{align} 
We observe that $\taunat(1g \cdot 1h) = \rhoreg_g \times \taunat(1h)$ where the multiplication on the right-hand side is a matrix-vector product. Consequently, we obtain the \emph{regular matrix representation}~\cite{jacobson2009basic} of the algebra $\Fb_2[G]$ 
\begin{align} \label{eq:reg_mat_rep}
\sum_{g \in G}r_g g \to \sum_{g \in G}r_g \rhoreg_g.
\end{align} 
% where the element \mbox{$r = \sum_{g \in G}r_g g \in \Fb_2[G]$} is represented by the matrix \mbox{$\sum_{g \in G}r_g \rho_g$}. 
This is an injective algebra homomorphism from $\Fb_2[G]$ into $\Fb_2^{n \times n}$. For any choice of $r=\sum_{g \in G} r_g g \in \Fb_2[G]$ and $m \in \Fb_2[G]$, we have
\begin{align}
\taunat(rm) &= \left( \sum_{g \in G} r_g \rhoreg_g \right) \times \taunat(m) \nonumber \\
&= \sum_{g \in G} r_g \, \left( \rhoreg_g \times \taunat(m) \right).
\label{eq:taunat_rhoreg}    
\end{align}

In particular, this is valid for any ring element $r \in \Fb_2[G]$ and module element $m \in \Cs$ since $\Cs \subset \Fb_2[G]$.

\begin{example}
To illustrate the matrix representation of a group algebra, consider the ring $\Fb_2[C_3]$ where $C_3 = \{e,\gamma,\gamma^2\}$ is the cyclic group of order $3$ and $e \in C_3$ is the identity element. 
Let $\taunat: \Fb_2[C_3] \to \Fb_2^3$ be the map $m_e e + m_{\gamma} \gamma + m_{\gamma^2} \gamma^2 \to (m_e,m_{\gamma},m_{\gamma^2})$. 
The matrix representation $\rhoreg_e$ of the identity element $e$ is the $3 \times 3$ identity matrix over $\Fb_2$ while
\begin{align*}
\rhoreg_{\gamma} = 
\begin{bmatrix}
0 & 0 & 1 \\
1 & 0 & 0 \\
0 & 1 & 0
\end{bmatrix}
\text{ and }
\rhoreg_{\gamma^2} = 
\begin{bmatrix}
0 & 1 & 0 \\
0 & 0 & 1 \\
1 & 0 & 0
\end{bmatrix}.
\end{align*} 
The matrix representation of an arbitrary ring element $r_e e + r_{\gamma} \gamma + r_{\gamma^2} \gamma^2 \in \Fb_2[C_3]$ is 
\begin{align*}
\begin{bmatrix}
r_e & r_{\gamma^2} & r_{\gamma} \\
r_{\gamma} & r_e & r_{\gamma^2} \\
r_{\gamma^2} & r_{\gamma} & r_e
\end{bmatrix}.
\end{align*} 
\end{example}

\subsection{Permute-and-Add Network Codes from Group Codes} \label{sub:sec:permute-and-add}

Consider any $\Fb_2[G]$-linear network code over $\Cs$ where $\Cs$ is a left-ideal of $\Fb_2[G]$, i.e., $\Cs$ is a group code. We will use $\Xc=\Fb_2^n$ as the edge alphabet and the embedding  $\taunat: \Cs \to \Fb_2^n$ to map module elements to the edge alphabet for communication between the nodes of the network.
The symbol $X_e$ transmitted on any network edge $e$ is a length-$n$ vector over $\Fb_2$.
As usual, let the set of encoding and decoding coefficients of this network be $\{k^{d,e}\}$ and $\{k^{d,i}\}$, respectively. We denote the expansions of these coefficients as 
\begin{align*}
k^{d,e} = \sum_{g \in G} k^{d,e}_{g} g  \text{ and } k^{d,i} = \sum_{g \in G} k^{d,i}_g g,
\end{align*} 
where $k^{d,e}_g, k^{d,i}_g \in \Fb_2$.
Using the fact that $\taunat$ is a $\Fb_2$-linear map, we observe that the coding operation~\eqref{eq:coding_operation} on a outgoing edge $e$ of a non-source node $v$ is 
\begin{align*}
X_e &= \taunat \left( \sum_{d \in \rmin(v)} k^{d,e} \, \taunat^{-1}\left(X_d\right) \right) \\
&= \sum_{d \in \rmin(v)} \taunat \left(k^{d,e} \, \taunat^{-1} \left(X_d\right) \right).
\end{align*} 
Using~\eqref{eq:taunat_rhoreg} here, we obtain
\begin{align}
X_e &= \sum_{d \in \rmin(v)} \sum_{g \in G} k^{d,e}_g \, \left(\rhoreg_g \times X_d \right) \nonumber \\
&= \sum_{d \in \rmin(v)} \sum_{\substack{g \in G:\, k^{d,e}_g = 1}} \rhoreg_g \times X_d. \label{eq:enc_permute_add}
\end{align} 
Similarly, to execute the decoding operation~\eqref{eq:decoding_operation} at a node $v$ that demands a message $Z_i$, we can retrieve $\taunat(Z_i)$ as 
\begin{align}
\taunat(Z_i) &= \taunat\left( \sum_{d \in \rmin(v)} k^{d,i} \taunat^{-1}\left(X_d\right) \right) \nonumber \\
&= \sum_{d \in \rmin(v)} \taunat\left(k^{d,i} \taunat^{-1}\left(X_d\right)\right) \nonumber \\
&= \sum_{d \in \rmin(v)} \sum_{g \in G: k^{d,i}_g=1} \rhoreg_g \times X_d. \label{eq:dec_permute_add}
\end{align} 

Thus, the encoding and decoding operations at the non-source nodes involves applying (possibly several) permutations $\rhoreg_g$ on each incoming vector $X_d$ and computing the sum of the resulting vectors. 
Hence, this is a permute-and-add network code.
These coding operations are simple since \emph{(i)}~arithmetic over large finite fields is not required, which is often the case with multicast networks; and \emph{(ii)}~the task of multiplying a vector by a permutation matrix can be effected by a mere reordering of the elements of the vector. 

We also note that this network code is a linear network code over the $\Fb_2[G]$-left module $\Cs$. The embeddings $\taunat$ and $\rhoreg$ simply allow us to realize the coding operations as sums of matrix-vector products, i.e., as a fractional linear network code~\cite{CDFZ_IT_06} over $\Fb_2$. Hence, we can use the framework of~\cite{CoZ_Part1_IT_18,CoZ_Part2_IT_18} to study the existence of network coding solutions.

\begin{remark}
The \emph{circular-shift network codes} proposed in~\cite{TSLYL_IT_19,STLYL_TCOMM_19} correspond to the case where $G$ is a cyclic group of odd order. Since $G$ is cyclic, the permutations $\rhoreg_g$, $g \in G$, used for encoding and decoding are all cyclic permutations. The odd order of the group implies that the characteristic of $\Fb_2$ does not divide $|G|$, and hence, the group algebra $\Fb_2[G]$ is semi-simple, i.e., it is isomorphic to a product of matrix rings over finite fields. This semi-simple structure can be used to determine the conditions for the existence of network coding solutions when this group algebra is used (see Section~\ref{sec:existence_of_solutions}).
\end{remark}

In the literature~\cite{TSLYL_IT_19,STLYL_TCOMM_19}, the complexity of a permute-and-add network code is measured in terms of the number of permutations applied on each incoming vector $X_d$.
Note that the number of permutations applied on $X_d$ in the encoding operation~\eqref{eq:enc_permute_add} is 
% \begin{equation*}
\mbox{$|\{g \in G~|~k^{d,e}_g = 1\}| = \wt(\taunat(k^{d,e}))$}
% \end{equation*} 
which is the Hamming weight of the vector $\taunat(k^{d,e})$. We abuse the notation mildly to denote this quantity as $\wt(k^{d,e})$.
Similarly, the number of permutations applied on $X_d$ in the decoding operation~\eqref{eq:dec_permute_add} is $\wt(k^{d,i})$.

\begin{definition}
An $\Fb_2[G]$-linear network code over a group code $\Cs$ is of \emph{degree} $\delta$ if the weight of each of its encoding and decoding coefficients is at the most $\delta$, i.e., $\wt(k^{d,e}), \wt(k^{d,i}) \leq \delta$ for all coefficients $k^{d,e}$ and $k^{d,i}$.
\end{definition}

% Note that the degree $\delta$ lies between $1$ and $n=|G|$.

We now apply the main principle Lemma~\ref{lem:annihilator} to modify, in a structured manner, a given network code over a group code and determine the degree of the modified network code. The annihilator $\ann(\Cs)$ of the left-module $\Cs$ is a two-sided ideal of $\Fb_2[G]$, see~\cite{Herstein_1968}. In particular, $\ann(\Cs)$ is an additive subgroup of $\Fb_2[G]$. Since, $\taunat$ is a $\Fb_2$-linear map we deduce that $\taunat(\ann(\Cs))=\{\taunat(a)~|~a \in \ann(\Cs)\}$ is a subgroup of $\Fb_2^n$, i.e., $\taunat(\ann(\Cs))$ is a binary linear code. The \emph{covering radius} of this linear code is
\begin{equation*}
\rcov = \max_{\pmb{v} \in \Fb_2^n} \,\left\{ \min_{\pmb{a} \in \taunat(\ann(\Cs))} \wt(\pmb{v} + \pmb{a}) \right\},
\end{equation*} 
which is the farthest that any of the vectors in $\Fb_2^n$ lies from $\taunat(\ann(\Cs))$. 
Thus, for any given $\pmb{v} \in \Fb_2^n$ there exists a vector in $\taunat(\ann(\Cs))$ at a distance of at the most $\rcov$. 
By abusing the notation mildly, we use $\rcov(\ann(\Cs))$ or simply $\rcov$ to denote the covering radius of $\taunat\left(\ann(\Cs)\right)$.

We modify a given network code $\{k^{d,e}\}, \{k^{d,i}\}$ as follows. For each coefficient $k^{d,e}$ we choose a vector $\pmb{a}^{d,e} \in \taunat(\ann(\Cs))$ such that $\wt\left(\taunat(k^{d,e}) + \pmb{a}^{d,e}\right) \leq \rcov$. We then use $a^{d,e} = \taunat^{-1}(\pmb{a}^{d,e})$ to modify the coefficient $k^{d,e}$ to $k^{d,e} + a^{d,e}$. Then \eqref{eq:enc_permute_add} is rewritten as
\begin{align}
X_e &=\sum_{d \in \rmin(v)} \sum_{\substack{g \in G:\\ k^{d,e}_g + a^{d,e}_g= 1}} \rhoreg_g \times X_d. \label{eq:modified_enc_permute_add}
\end{align} 
where, $\pmb{a}^{d,e}=(a^{d,e}_g)_{g \in G}\in \Fb_2^n $.

We observe that $a^{d,e} \in \ann(\Cs)$ and the weight of the modified coefficient $k^{d,e} + a^{d,e}$ is at the most $\rcov$. Therefore, from \eqref{eq:modified_enc_permute_add} it is obvious that number of permutations on any incoming vector of an intermediate node $v$ is at the most $\rcov$.
Using a similar strategy, for each decoding coefficient $k^{d,i}$ we choose $a^{d,i} \in \ann(\Cs)$ such that $\wt(k^{d,i} + a^{d,i}) \leq \rcov$ and hence, number permutations on any incoming vector of a sink node $v$ is also at the most $\rcov$. So we conclude that the modified network code is of degree $\rcov$. 

Using Lemma~\ref{lem:annihilator}, we see that any $\Fb_2[G]$-linear network code over a group code $\Cs$ can be modified to a degree-$\rcov$ network code without affecting the messages passed in any of the edges or the ability of the sinks to decode their demands. 
% Hence, without loss of generality, we identify a given network code with its modified counterpart. 
 In summary, suppose $\{k^{d,e}\}, \{k^{d,i}\}$ is a network code over the ideal $\Cs$. We first modify these coefficients by adding elements $a^{d,e},a^{d,i} \in \ann(\Cs)$ so that the weights of the resulting coefficients is at the most $\rcov(\ann(\Cs))$. We then obtain a vector linear network code for the network, which is a permute-and-add network code with coding coefficients $\left\{\pmb{K}^{d,e}=\sum_{g \in G: k^{d,e}_g + a^{d,e}_g = 1} \rhoreg_g\right\}$ and $\left\{\pmb{K}^{d,i}=\sum_{g \in G: k^{d,i}_g + a^{d,i}_g = 1} \rhoreg_g\right\}$. This vector linear network code is a solution to this network if and only if the original network code over $\Cs$ is a solution.
By doing so we arrive at the main result of this section.

\begin{theorem} \label{thm:degree}
% Any $\Fb_2[G]$-linear network code over a left-ideal $\Cs$ is a degree $\rcov(\ann(\Cs))$ permute-and-add network code.
Let $\Cs$ be a left-ideal in $\Fb_2[G]$. Any linear network code over $\Cs$ for a network yields a permute-and-add network code for this network with degree equal to the covering radius of $\ann(\Cs)$.
\end{theorem}

Since we use $\Fb_2^n$ as the edge alphabet, the rate of a network code over $\prescript{}{\Fb_2[G]}{\Cs}$ is $\log_2 |\Cs| / \log_2 |\Fb_2^n| = \dim(\Cs)/n$, which is the ratio of the dimension of $\Cs$ (as a vector space over $\Fb_2$) to the order of the group $G$.
If $\Cs$ is non-trivial, i.e., $\Cs \neq \Fb_2[G]$, then the rate is less than $1$. 

\begin{remark} \label{rem:trade-off}
There exists a trade-off between rate and degree. If ideals $\Cs'$ and $\Cs$ are such that \mbox{$\Cs' \subset \Cs$}, then \mbox{$\ann(\Cs') \supset \ann(\Cs)$}, and hence, $\rcov(\ann(\Cs')) \leq \rcov(\ann(\Cs))$. Thus a network code designed over a smaller ideal will achieve a smaller degree at the cost of yielding a lower rate. See Example~\ref{ex:rate-degree_tradeoff} (Section~\ref{sec:existence_of_solutions}) for an illustration.
\end{remark}

As a corollary to Theorem~\ref{thm:degree} we have the following upper bound on the degree for a wide class of network codes.

\begin{corollary} \label{cor:even_weight}
If $\Cs$ is a left-ideal of $\Fb_2[G]$ such that the weight of every element of $\Cs$ is even, then the degree of the permute-and-add network code corresponding to any linear network code over $\Cs$ is at the most $\left\lfloor \frac{n}{2}\right\rfloor$.
\end{corollary}

\begin{proof}
We first observe that $\sum_{h \in G}1h \in \ann(\Cs)$, since for any $\sum_{g \in G}m_g g \in \Cs$ we have
\begin{equation*}
\left(\sum_{h \in G}1h\right) \cdot \left(\sum_{g \in G}m_g\right) = \sum_{g \in G} \left(\sum_{h \in G} m_h\right) g = 0.
\end{equation*} 
Clearly, $\sum_{g \in G} 0g \in \ann(\Cs)$. Thus, $\taunat(\ann(\Cs))$ contains the all-zeros and the all-ones vectors, i.e., the repetition code is a subcode of $\taunat(\ann(\Cs))$. Since the covering radius of the repetition code is $\left\lfloor \frac{n}{2} \right\rfloor$, we conclude that $\rcov(\ann(\Cs)) \leq \left\lfloor \frac{n}{2} \right\rfloor$. This completes the proof.
\end{proof}

\begin{remark} \label{rem:bit_truncation}
A process called \emph{bit truncation} was used in conjunction with circular-shift network codes for combination networks in~\cite{XMA_ISIT07} which reduces the length of the edge alphabet vector from $n$ to $n-1$, and thereby increases the rate. A similar process can be used in our setting when $\Cs$ contains only even weight elements at the cost of a marginal increase in complexity. 
If $\Cs$ contains only even weight elements, then a module element $m \in \Cs$ carried on any edge of the network has even weight. Hence, the first $n-1$ components of the $n$-length vector $\taunat(m)$ are sufficient to determine $\taunat(m)$ since the last component can be obtained as the $\Fb_2$-sum of the first $n-1$ components. This technique uses $\Fb_2^{n-1}$ instead of $\Fb_2^n$ as the edge alphabet, and hence, increases the rate of the network code from $\dim(\Cs)/n$ to $\dim(\Cs)/(n-1)$. However, this involves additional complexity due to the effort required to compute the parity of the $(n-1)$ bits transmitted on each edge of the network.
\end{remark}

Note that Theorem~\ref{thm:degree} does not address the question of whether a linear network coding \emph{solution} over the left-ideal $\Cs$ exists. In Section~\ref{sec:existence_of_solutions} we will rely on the framework developed in~\cite{CoZ_Part2_IT_18} to discuss the existence of network coding solutions over $\Cs$.
% when $\Fb_2[G]$ is semi-simple.

\section{Existence of Network Coding Solutions} \label{sec:existence_of_solutions}

We will first recall some basic results related to network coding over rings~\cite{CoZ_Part2_IT_18} in Section~\ref{sub:sec:review_network_coding_over_rings}, and then use them to analyze network codes over ideals of $\Fb_2[G]$ in the rest of this section. In particular, Section~\ref{sub:sec:semi-simple} provides necessary and sufficient conditions for existence of network codes when $\Fb_2[G]$ is a commutative semi-simple ring, Section~\ref{sub:sec:circular-shift} analyzes circular-shift network codes, and Section~\ref{sub:sec:non-cyclic-abelian} presents an example where $G$ is a non-cyclic group.

\subsection{Review of Linear Network Codes over Modules} \label{sub:sec:review_network_coding_over_rings}

Lemma~I.6 of~\cite{CoZ_Part2_IT_18} is a key result that connects the question of linear solvability of a network over two different rings. % that are defined over different rings.

\begin{theorem} \cite[Lemma~I.6]{CoZ_Part2_IT_18} \label{thm:CoZ_1_6}
Let the rings $R$ and $S$ be such that there exists a ring homomorphism from $R$ to $S$. If a network is solvable over some faithful $R$-module then it is solvable over every $S$-module.
\end{theorem} 

Every ring is a faithful module over itself. Hence, if a network is scalar linearly solvable over $R$ (i.e., over the module $\prescript{}{R}{R}$) and if there is a homomorphism from $R$ to $S$, then the network is scalar linearly solvable over $S$.
Further, choosing $S=R$ and the identity map as the homomorphism, we observe that linear solvability over some faithful $R$-module implies scalar linear solvability over $R$.

As a corollary to Theorem~\ref{thm:CoZ_1_6} we obtain a generalization of a key observation from~\cite{STLYL_TCOMM_19} regarding circular-shift linear network codes. In the terminology of network coding over rings this observation from~\cite{STLYL_TCOMM_19} states that for any finite \emph{cyclic} group $G$, a multicast network has a scalar linear solution over $\Fb_2[G]$ if and only if it has a scalar linear solution over $\Fb_2$. 
Note that a scalar linear solution over $\Fb_2[G]$ uses $\Cs=\Fb_2[G]$ as the module, and hence, achieves the maximum possible rate $\dim(\Cs)/|G| = 1$.
This maximum possible rate can be achieved by circular-shift network codes only if the network has a scalar linear solution over $\Fb_2$, i.e., if the network already enjoys a low complexity solution. For such networks, it might be unnecessary to use circular-shift network codes (instead of solutions over $\Fb_2$).  
Thus, while designing circular-shift network codes the networks that are of interest are those that do not have a scalar linear solution over $\Fb_2$, and for these networks a solution exists only if $\Cs \subsetneq \Fb_2[G]$, i.e., the achievable rate is strictly less than $1$.
The following result is a generalization to arbitrary finite groups $G$ and arbitrary networks (not necessarily multicast).

\begin{corollary} \label{cor:existence_over_F2}
Let $G$ be a finite group. A network has a scalar linear solution over $\Fb_2[G]$ if and only if it has a scalar linear solution over $\Fb_2$.
\end{corollary}
\begin{proof}
The function that maps $\sum_{g} a_g g$ to $\sum_g a_g$ is a ring homomorphism from $\Fb_2[G]$ onto $\Fb_2$. From Theorem~\ref{thm:CoZ_1_6}, if there is a scalar linear solution over $\Fb_2[G]$ (which is a faithful $\Fb_2[G]$-module), there exists a scalar linear solution over $\Fb_2$. 

Similarly, the function that maps $a \in \Fb_2$ to $a\,e \in \Fb_2[G]$, where $e$ is the identity element of $G$, is a ring homomorphism. Thus, the existence of a scalar linear solution over $\Fb_2$ implies the existence of a scalar linear solution over $\Fb_2[G]$.
\end{proof}

Lemma~II.6 of~\cite{CoZ_Part2_IT_18} analyzes the case where a network is solvable over an unfaithful $R$-module $M$. We review the aspects of the proof of this lemma that are important for us. The proof uses the fact that $\ann(M)$ is a two-sided ideal in $R$ and that $M$ is a faithful $R/\ann(M)$-module, see~\cite{Herstein_1968}. Using the natural homomorphism from $R$ to $R/\ann(M)$ the proof shows that the existence of an $\prescript{}{R}{M}$ linear solution implies the existence of an $\prescript{}{R/\ann(M)}{M}$ linear solution. Since $M$ is faithful over $R/\ann(M)$, using Theorem~\ref{thm:CoZ_1_6}, we conclude that the existence of an $\prescript{}{R}{M}$-linear solution implies the existence of a scalar linear solution over $R/\ann(M)$. 
% Thus, we rephrase Lemma~II.6 of~\cite{CoZ_Part2_IT_18} as follows. 
Hence, the statement of~\cite[Lemma~II.6]{CoZ_Part2_IT_18} is essentially the ``only if'' part of % the following result.

\begin{theorem} \label{thm:CoZ_2_6} % \cite[Lemma~II.6]{CoZ_Part2_IT_18} 
A network is linearly solvable over $\prescript{}{R}{M}$ if and only if it is scalar linearly solvable over $R/\ann(M)$.
\end{theorem}
\begin{proof}
The proof of the ``if'' part is similar to the proof of~\cite[Lemma~II.6]{CoZ_Part2_IT_18} and uses the same ideas in the logically reverse direction. We provide a detailed proof here for completeness. 

Assume that $M$ is an $R$-module and denote the annihilator of $M$ over $R$ as $J=\ann(M)$. Let $S=R/J$ be the ring whose elements are the distinct cosets $r + J$ of $J$ in $R$. Extend the action of $R$ on $M$ to an action $\odot: S \times M \to M$ defined as 
% \begin{align*}
$s \odot m = rm$, for any choice of $r \in R$ belonging to the coset $s \in R/J$.
% \end{align*} 
Note that this a valid definition, since for any two elements $r,r' \in s$, we have $r-r' \in J$, and hence, $(r-r')m =0$, i.e., $rm=r'm$.
It is straightforward to verify that $\odot$ is an action that makes $M$ an $S$-module.

Now assume that a given network has a scalar linear solution over $S$, i.e., over the module $\prescript{}{S}{S}$. Since this is a faithful $S$-module, from Theorem~\ref{thm:CoZ_1_6}, this network has a solution over $\prescript{}{S}{M}$. Let the coding coefficients of this solution be $\{\bar{k}^{d,e}\}$, $\{\bar{k}^{d,i}\}$, where each coefficient is a coset of $J$ in $R$. In order to construct a network code over $\prescript{}{R}{M}$, we choose the network coding coefficients $\{k^{d,e}\}$, $\{k^{d,i}\}$ as follows: let \mbox{$k^{d,e} \in R$} be any element of the coset $\bar{k}^{d,e}$, and let $k^{d,i} \in R$ be any element of the coset $\bar{k}^{d,i}$. 
With this choice, we know that $k^{d,e}m = \bar{k}^{d,e} \odot m$ and $k^{d,i}m = \bar{k}^{d,i} \odot m$ for all $m \in M$. That is, the action of the network coding coefficients from $R$ and the action of the corresponding coefficients from $S$ are identical when applied on $M$. 
A standard induction argument on the topologically ordered list of vertices of the directed acyclic network (as used in the proofs of Lemma~\ref{lem:annihilator} and~\cite[Lemma~II.6]{CoZ_Part2_IT_18}) shows that for each edge of the network the symbol carried by the $\prescript{}{S}{M}$-linear network code and the symbol carried by the $\prescript{}{R}{M}$-linear network code are identical, and hence, $\{k^{d,e}\}$, $\{k^{d,i}\}$ is a linear solution over $\prescript{}{R}{M}$.
\end{proof}

\subsection{Network Codes from Semi-Simple Abelian Group Algebras} \label{sub:sec:semi-simple}

\emph{In the rest of this section we will assume that $G$ is an Abelian group of \emph{odd} order and $\Cs$ is an ideal in the commutative ring $\Fb_2[G]$.}
% We will refer to $\Cs$ as an Abelian code.

The fact that $2 \nmid |G|$ implies that $\Fb_2[G]$ is semi-simple, i.e., $\Fb_2[G]$ is isomorphic to a direct a product of matrix rings over finite fields, and the assumption that $G$ is Abelian implies that the ring $\Fb_2[G]$ is commutative and each matrix ring in its decomposition is a finite field of characteristic $2$. 
Several well known families of error correcting codes are ideals in Abelian group algebras, such as BCH codes, punctured Reed-Muller codes, quadratic residue codes and bicyclic codes~\cite{Ber_Cybernetics_II_77,Wil_BellSys_70,Imai_InfControl_77,KeS_ContemporaryMath_01}.

The transform domain treatment of Abelian codes in~\cite{RaS_IT_92} provides an isomorphism of $\Fb_2[G]$ onto a direct product of finite fields using Discrete Fourier Transforms. This subsumes the spectral characterization of cyclic codes~\cite{blahut2003algebraic}. We know that~\cite[Theorem~1]{RaS_IT_92} there exists an isomorphism 
\begin{equation} \label{eq:dft_isomorphism}
\dft:~ \Fb_2[G] \to \mathcal{R} \triangleq \Fb_{q_1} \times \Fb_{q_2} \times \cdots \times \Fb_{q_t}
\end{equation} 
where $t$ is a positive integer and $q_1,\dots,q_t$ are powers of $2$. The ring $\mathcal{R}$ has $t$ minimal ideals, the $k^\tth$ minimal ideal 
% \begin{equation*}
% \{0\} \times \cdots \times \{0\} \times \Fb_{q_k} \times \{0\} \times \cdots \times \{0\}
% \end{equation*} 
is generated by $\theta_k=(0,\dots,0,1,0,\dots,0)$ where the only non-zero entry of $\theta_k$ occurs in the $k^\tth$ position. The ideal generated by $\theta_k$ is 
\begin{equation*}
\langle \theta_k \rangle = \{0\} \times \cdots \times \{0\} \times \Fb_{q_k} \times \{0\} \times \cdots \times \{0\}.
\end{equation*} 
The ring $\mathcal{R}$ contains exactly $2^t$ ideals, and any ideal of $\mathcal{R}$ is a direct sum of some of the minimal ideals. If $J$ is an ideal of $\mathcal{R}$ then there exists a $T(J) \subset \{1,2,\dots,t\}$ such that $J = \oplus_{k \in T(J)} \langle \theta_k \rangle$.
It is straightforward to show that
\begin{align}
\ann(J) &= \oplus_{k \notin T(J)} \left\langle \theta_k \right\rangle \text{ and }
\mathcal{R}/\ann(J) \cong \prod_{k \in T(J)} \Fb_{q_k}. \label{eq:R_by_annJ}
\end{align} 
If $\Cs$ is an ideal in $\Fb_2[G]$ and $J=\dft(\Cs)$ is the image of $\Cs$ in $\mathcal{R}$ under the isomorphism~\eqref{eq:dft_isomorphism}, then we will use $T(\Cs)$ to denote $T(J)$. 

We are now ready to characterize the existence of network coding solutions over Abelian codes $\Cs$. 
% {\color{blue}The following lemma gives the necessary and sufficient condition for the existence of a network coding solution over $\Cs$.}

\begin{lemma} \label{lem:solution_condition_semisimple}
Let $\Cs$ be an ideal in the semi-simple commutative group algebra $\Fb_2[G]$. A network has a linear solution over $\prescript{}{\Fb_2[G]}{\Cs}$ if and only if the network is scalar linearly solvable over each finite field $\Fb_{q_k}$, $k \in T(\Cs)$.
\end{lemma}
\begin{proof}
From Theorem~\ref{thm:CoZ_2_6},~\eqref{eq:dft_isomorphism} and~\eqref{eq:R_by_annJ}, a network is solvable over the $\Fb_2[G]$-module $\Cs$ if and only if the network has a scalar linear solution over 
\begin{equation*}
\Fb_2[G]/\ann(\Cs) \cong \mathcal{R}/\dft(\ann(\Cs)) \cong \prod_{k \in T(\Cs)} \Fb_{q_k}.
\end{equation*} 
From Lemma~II.12 of~\cite{CoZ_Part1_IT_18} we know that a network is scalar linearly solvable over a finite direct product of finite rings if and only if it is solvable over each ring in the product. This completes the proof.
\end{proof}

The number of finite fields in the decomposition~\eqref{eq:dft_isomorphism} and the sizes of these finite fields can be determined from the \emph{conjugacy classes} of $G$~\cite{RaS_IT_92}. 
We now recall this result from~\cite{RaS_IT_92}.
As usual, we will assume that $(G,\cdot)$ is a multiplicative group. The conjugacy class $C_g$ containing the group element $g \in G$ is
\begin{equation*}
C_g = \left\{g,g^2,g^4,\dots,g^{2^{\ell-1}}\right\}
\end{equation*} 
where $\ell=|C_g|$ is the smallest integer such that $g^{2^\ell}=g$, and is known as the \emph{exponent} of $C_g$. 
The distinct conjugacy classes of $G$ form a partition of $G$. The number finite fields $t$ in the decomposition~\eqref{eq:dft_isomorphism} is equal to the number of distinct conjugacy classes of $G$. 
Let $g_1,\dots,g_t \in G$ be such that $C_{g_1},\dots,C_{g_t}$ are the distinct conjugacy classes. Then $G = C_{g_1} \cup \cdots \cup C_{g_t}$ and 
\begin{equation*}
\textstyle \Fb_2[G] \cong \prod_{k=1}^{t} \Fb_{q_k} \text{ where } q_k = 2^{|C_{g_k}|} \text{ for each } k=1,\dots,t.
\end{equation*} 

\subsection{Circular-Shift Linear Network Codes} \label{sub:sec:circular-shift}

We now retrieve the main algebraic results on circular-shift linear network codes derived in~\cite{TSLYL_IT_19,STLYL_TCOMM_19} using our group algebraic framework. 

Circular-shift linear network codes correspond to the case where $G=\left\{e,y,y^2,\dots,y^{n-1}\right\}$, $y^n=e$, is a cyclic group. This is because the regular matrix representation of $g \in G$ is the $n \times n$ cyclic permutation matrix $\rhoreg_g$ over $\Fb_2$.
We represent the elements of group algebra $\Fb_2[G]$ as $\sum_{i=0}^{n-1} m_i y^i$.
Let the distinct conjugacy classes be those generated by the elements $y^{j_1},\dots,y^{j_t}$, i.e., $C_{y^{j_1}},\dots,C_{y^{j_t}}$ where \mbox{$j_1,\dots,j_t \in \left\{0,1,\dots,n-1\right\}$}. % and let their exponents be $\ell_1,\dots,\ell_t$, respectively.
Further, let $\omega$ be a primitive $n^\tth$ root of unity in a suitable algebraic extension of $\Fb_2$.
The map $\Phi\left(\sum_{i=0}^{n-1} m_i y^{i}\right) = \left(\hat{m}_1,\dots,\hat{m}_t\right)$ where
\begin{equation*}
 \hat{m}_k = \sum_{i=0}^{n-1} m_i \omega^{i j_k} \text{ for } k=1,\dots,t,
\end{equation*} 
is an isomorphism~\eqref{eq:dft_isomorphism}. 
Observe that $|C_{y^{j_k}}|$ is the smallest integer $\ell_k$ such that $y^{j_k 2^{\ell_k}} = y^{j_k}$, i.e., $j_k 2^{\ell_k} = j_k \mod n$. And the $k^\tth$ component in the finite field decomposition~\eqref{eq:dft_isomorphism} is $\Fb_{2^{\ell_k}}$.

We will use the convention that $j_1=0$, i.e., $C_{y^{j_1}} = C_{e} = \{e\}$. The exponent of this conjugacy class is $1$, and thus the first finite field in the decomposition~\eqref{eq:dft_isomorphism} is $\Fb_{q_1}=\Fb_2$. The corresponding minimal ideal is $\langle \theta_1 \rangle = \Fb_2 \times \{0\} \times \cdots \times \{0\}$.

If a group code $\Cs$ is such that $\Phi(\Cs) \supset \langle \theta_1 \rangle$ then $1 \in T(\Cs)$. 
For any such $\Cs$, from Lemma~\ref{lem:solution_condition_semisimple}, a network has a solution over $\Cs$ only if it is scalar linearly solvable over $\Fb_2$. 
Also, any other finite field $\Fb_{q_k}$ in the decomposition~\eqref{eq:dft_isomorphism} is a field extension of $\Fb_2$, i.e., there exists a ring homomorphism from $\Fb_2$ to $\Fb_{q_k}$. Hence, from Theorem~\ref{thm:CoZ_1_6}, if a network is scalar linearly solvable over $\Fb_2$ then it is scalar linearly solvable over each $\Fb_{q_k}$, $k \in T(\Cs)$.
Using these observations with Lemma~\ref{lem:solution_condition_semisimple} we conclude that if $\Cs$ is such that $1 \in T(\Cs)$ then a network is solvable over $\Cs$ if and only if it is scalar linearly solvable over $\Fb_2$.

% Hence, the network coding problems that are of interest to us correspond to the cases where $\langle \theta_1 \rangle \not\subset \Cs$, or equivalently, $1 \notin T(\Cs)$.
% For any such $\Cs$, from Lemma~\ref{lem:solution_condition_semisimple}, a network has a solution over $\Cs$ if and only if it is scalar linearly solvable over $\Fb_2$. Since a scalar linear solution has the least possible complexity, circular-shift linear network codes might not be necessary for such networks. 
% Hence, the network coding problems that are of interest to us correspond to the cases where $\langle \theta_1 \rangle \not\subset \Cs$, or equivalently, $1 \notin T(\Cs)$.

Now, let $\Cs$ be such that \mbox{$1 \notin T(\Cs)$}. This implies that for any \mbox{$\sum_{i=0}^{n-1}m_i y^i \in \Cs$}, the image $\Phi\left(\sum_{i=0}^{n-1} m_i y^i\right) = (\hat{m}_1,\hat{m}_2,\dots,\hat{m}_t)$ satisfies 
% \begin{equation*}
\mbox{$0 = \hat{m}_1 = \sum_{i=0}^{n-1} m_i$}.
% \end{equation*} 
Hence, every element of $\Cs$ has even weight. From Corollary~\ref{cor:even_weight} we deduce that any network code over $\Cs$ is of degree at the most \mbox{$(n-1)/2$}. Since $G$ is a cyclic group, this network code over $\Cs$ corresponds to a circular shift network code of degree at the most \mbox{$(n-1)/2$}. Hence, we have proved

\begin{lemma} \label{lem:n-1_by_2_circular}
Let $G$ be a cyclic group of odd order $n$, and $\Cs$ be any ideal of $\Fb_2[G]$ such that $1 \notin T(\Cs)$. The degree of a circular shift  network code obtained from any network code over $\Cs$ is at the most $(n-1)/2$.
\end{lemma}

\subsubsection{When $n$ is prime with primitive root $2$} 

In this case $2$ is a primitive root modulo $n$, with $n$ prime, i.e., every integer $1,\dots,n-1$ is some power of $2$ modulo $n$.
Thus, the conjugacy class $C_{y}$ is equal to $\left\{y,y^2,y^3,\dots,y^{n-1}\right\}$. The remaining element of the group is the identity which forms a conjugacy class by itself $C_e = \{e\}$. The exponents of the conjugacy classes $C_e$ and $C_y$ are $1$ and $n-1$, respectively.
Hence,
\begin{equation*}
\Fb_2[G] \cong \Fb_2 \times \Fb_{2^{n-1}}.
\end{equation*} 
Let $\Cs$ be the ideal $\Phi^{-1}\left(\left\{0\right\} \times \Fb_{2^{n-1}}\right)$. 
% For any $\sum_{i=0}^{n-1}m_i y_i \in \Cs$, the image $\Phi(\sum_{i=0}^{n-1} m_i y_i) = (\hat{m}_1,\hat{m}_2)$ satisfies 
% \begin{equation*}
%  0 = \hat{m}_1 = \sum_{i=0}^{n-1} m_i.
% \end{equation*} 
% Hence, every element of $\Cs$ has even weight. 
From Lemma~\ref{lem:n-1_by_2_circular} we deduce that any network code over $\Cs$ is of degree \mbox{$(n-1)/2$}. The rate of any network code over $\Cs$ is \mbox{$\log_2 |\Cs|/n = (n-1)/n$.} Further, from Lemma~\ref{lem:solution_condition_semisimple}, a network has a solution over $\Cs$ if and only if it has a scalar linear solution over $\Fb_{2^{n-1}}$. This result generalizes~\cite[Theorem~4]{TSLYL_IT_19} from multicast networks to arbitrary networks.

We also recover the main result of \emph{rotate-and-add} network coding~\cite{HaK_ITW_10} for single source multicast networks. A rotate-and-add network code is a linear network code that uses only cyclic permutations as encoding kernels in the intermediate nodes of a network, i.e., \mbox{$\wt(k^{d,e})=1$} for all adjacent pairs of edges $(d,e)$ in the network. 
In other words, the coding operation~\eqref{eq:enc_permute_add} of a rotate-and-add network code involves the application of exactly one cyclic permutation on each incoming vector $X_d$.
However, there are no restrictions imposed on the decoding coefficients $k^{d,i}$.
The following is a restatement of one of the key results (with a mild strengthening) on rotate-and-add network codes using our terminology.

\begin{lemma}{~\cite[Theorem~1]{HaK_ITW_10}} \label{lem:rotate-and-add}
Let $n$ be an odd prime with primitive root $2$, and $G$ be the cyclic group of order $n$. A multicast network with $N_{\sf rx} \leq n$ sinks has a $\Fb_2[G]$-linear network coding solution with rate $(n-1)/n$ where all encoding coefficients $k^{d,e}$ satisfy $\wt(k^{d,e})=1$ and all decoding coefficients $k^{d,i}$ satisfy $\wt(k^{d,i}) \leq (n-1)/2$.
\end{lemma}
\begin{proof}
We know that $\Fb_2[G] \cong \Fb_2 \times \Fb_{2^{n-1}}$. The isomorphism~\eqref{eq:dft_isomorphism} $\Phi\left(\sum_{i=0}^{n-1} m_i y^i\right) = (\hat{m}_1,\hat{m}_2)$ satisfies $\hat{m}_1=\sum_i m_i$ and $\hat{m}_2 = \sum_{i} m_i \alpha^i$ where $\alpha$ is a primitive $n^\tth$ root of unity~\cite{blahut2003algebraic,RaS_IT_92}.
When $m=1e, 1y, 1y^2,\dots, 1y^{n-1}$, the corresponding image $(\hat{m}_1,\hat{m}_2)$ satisfies $\hat{m}_2=1,\alpha,\alpha^2,\dots,\alpha^{n-1}$, respectively, and $\hat{m}_1=1$.

From the Jaggi-Sanders algorithm~\cite{JSCEEJT_IT_05} we know that a multicast network has a scalar linear solution over a finite field $\Fb_q$ with the encoding coefficients restricted to a subset $\mathcal{P} \subset \Fb_q$ if $N_{\sf rx} \leq |\mathcal{P}|$. Using the finite field $\Fb_{2^{n-1}}$ and the subset $\mathcal{P}=\left\{1,\alpha,\dots,\alpha^{n-1}\right\}$, we conclude that the given multicast network has a scalar linear solution over $\Fb_{2^{n-1}}$ with the encoding coefficients taking values from $\left\{1,\alpha,\cdots,\alpha^{n-1}\right\}$. We will denote the set of encoding coefficients of this scalar linear network code as $\left\{\bar{k}^{d,e}\right\}$.

We will use $\Cs=\{0\} \times \Fb_{2^{n-1}}$ to design our $\Fb_2[G]$-linear network code. In this case $\ann(\Cs)=\Fb_2 \times \{0\}$ and $\Fb_2[G] /\ann(\Cs) \cong \Fb_{2^{n-1}}$. From Theorem~\ref{thm:CoZ_2_6} and using the fact that $\{\bar{k}^{d,e}\}$ is a $\Fb_{2^{n-1}}$-linear solution, it is clear that the there exists a network coding solution over $\Cs$. Retracing the proof of Theorem~\ref{thm:CoZ_2_6}, we obtain a network code $\left\{k^{d,e}\right\}$ over $\Cs$ by choosing 
\begin{equation*}
 k^{d,e} = 1.y^i \text{ if } \bar{k}^{d,e} = \alpha^i.
\end{equation*} 
Clearly $\Phi\left(k^{d,e}\right) = (1,\bar{k}^{d,e})$ is an element of the coset of $\ann(\Cs)$ in $\Fb_2[G]$ to which $\Phi^{-1}\left((0, k^{d,e})\right)$ belongs. We conclude that this derived network code is a solution of rate $\dim(\Cs)/n=(n-1)/n$ and all encoding coefficients have weight $1$.

We further observe that the ideal $\Cs$ is such that $1 \notin T(\Cs)$. Hence, from Lemma~\ref{lem:n-1_by_2_circular}, the decoding coefficients $k^{d,i}$ can be chosen such that $\wt(k^{d,i}) \leq (n-1)/2$. 
\end{proof}

\subsubsection{When $n$ is an arbitrary odd integer}

Let $\ell_0$ be the multiplicative order of $2$ modulo $n$. Then $C_{y} = \left\{y,y^2,\dots,y^{2^{\ell_0 - 1}}\right\}$ and $|C_{y}|= \ell_0$. 
Assume, without loss of generality, that in the decomposition~\eqref{eq:dft_isomorphism} $q_1=2$ (as usual) and $q_2,\dots,q_{t_0+1}=2^{\ell_0}$. In other words, let $t_0$ be the number of conjugacy classes with exponent equal to $\ell_0$. We now derive a lower bound on the value of $t_0$.

\begin{lemma} \label{lem:circular_n_odd}
Let $\varphi(n) = \big|\left\{ j~|~(j,n)=1, 1 \leq j \leq n-1 \right\}\big|$ be the Euler's totient function. Then $\ell_0 | \varphi(n)$ and $t_0 \geq \varphi(n)/\ell_0$.
\end{lemma}
\begin{proof}
Consider any $j \in \{1,\dots,n-1\}$ with $(j,n)=1$. We will first show that $|C_{y^j}|=\ell_0$. Let $\ell$ be the smallest integer such that \mbox{$y^{j 2^\ell} = y^j$}, i.e., \mbox{$j 2^\ell = j \mod n$}. This implies that $n \big|~ j \left(2^\ell - 1\right)$. Since $(j,n)=1$, we deduce that $\ell$ is the smallest integer such that $n \big | \left(2^\ell-1\right)$. Note that this criterion is independent of $j$ as long as $(j,n)=1$. Hence, $|C_{y^j}| = |C_{y^1}| = \ell_0$.

Since $n$ is odd, observe that if $(j,n)=1$ then \mbox{$(2j \mod n,n)=1$} as well. Thus, $\{y^j~|~(j,n)=1\}$ is a disjoint union of conjugacy classes. From the discussion in the previous paragraph each such conjugacy class has size $\ell_0$. Thus, the size of this set $\left\{y^j~|~(j,n)=1\right\}$, viz. $\varphi(n)$, is divisible by $\ell_0$.
The set $\{y^j~|~(j,n)=1\}$ is a union of $\varphi(n)/\ell_0$ distinct conjugacy classes, each of size $\ell_0$. Hence, there exist at least $\varphi(n)/\ell_0$ conjugacy classes of $G$ whose exponent is equal to $\ell_0$. This implies that in the decomposition of $\Fb_2[G]$ as product of finite rings there are at least $\varphi(n)/\ell_0$ fields with size equal to $2^{\ell_0}$.
\end{proof}

Let $\Cs = \oplus_{k=2}^{t_0+1}\langle \theta_k \rangle$ be the direct sum of ideals corresponding to the $t_0$ conjugacy classes with exponent equal to $\ell_0$, i.e.,
\begin{equation} \label{eq:ideal_n_odd}
\Phi(\Cs) = \{0\} \times \Fb_{2^{\ell_0}} \times \cdots \times \Fb_{2^{\ell_0}} \times \{0\} \cdots \times \{0\}.
\end{equation} 
Since $1 \notin T(\Cs)$, from Lemma~\ref{lem:n-1_by_2_circular} we deduce that any network code over $\Cs$ is of degree $\rcov\left(\ann(\Cs)\right) \leq (n-1)/2$. To compute the rate, note that $\log_2 |\Cs| = t_0 \ell_0$. Applying Lemma~\ref{lem:circular_n_odd}, we have $\log_2 |\Cs| \geq \varphi(n)$, and hence, the rate of any network code over $\Cs$ is $t_0\ell_0/n \geq \varphi(n)/n$. Finally, we use Lemma~\ref{lem:solution_condition_semisimple} to see that a network has a solution over $\Cs$ if and only if it is scalar linearly solvable over $\Fb_{2^{\ell_0}}$.
Hence, we have proved

\begin{lemma} \label{lem:main_n_odd}
Let $\ell_0$ be the multiplicative order of $2$ modulo $n$, $t_0$ the number of $\ell_0$-sized conjugacy classes of the cyclic group of order $n$, and $\Cs$ be the ideal in~\eqref{eq:ideal_n_odd}.
If a network has a scalar linear solution over $\Fb_{2^{\ell_0}}$ then it has a degree $\rcov(\ann(\Cs)) \leq (n-1)/2$ circular-shift network coding solution with rate $t_0\ell_0/n \geq \varphi(n)/n$.
\end{lemma}

In general, the bounds promised by Lemma~\ref{lem:main_n_odd} on rate and degree are loose. We illustrate this in the following example.

\begin{example} \label{ex:n_15}
Let $n=15$. The following are the conjugacy classes of the cyclic group of order $15$,
\begin{align*}
&\left\{e\right\}, \left\{y,y^2,y^4,y^8\right\}, \left\{y^3,y^6,y^{12},y^9\right\}, \left\{y^7,y^{14},y^{13},y^{11}\right\}\\&\text{ and } \left\{y^5,y^{10}\right\}.
\end{align*}

In this case $\ell_0=4$, $\varphi(n)=8$, and hence, Lemma~\ref{lem:main_n_odd} guarantees the existence of a rate $8/15$ circular-shift network code of degree at the most $7$. This construction provides a solution if the network is scalar linearly solvable over $\Fb_{2^4}$.

On the other hand, we observe that $t_0=3$, and hence, we could use a group code $\Cs$ with $\Phi(\Cs)=\{0\} \times \Fb_{2^4} \times \Fb_{2^4} \times \Fb_{2^4} \times \{0\}$. This network code has rate $12/15$. We show in Appendix~\ref{app:ex:n_15} that $\rcov(\ann(\Cs))=6$, hence, this is a degree $6$ network code. A solution over $\Cs$ exists if and only if the network is scalar linearly solvable over $\Fb_{2^4}$.
\end{example}

The generality of our result in Lemma~\ref{lem:solution_condition_semisimple} provides design flexibility to trade-off rate for lower degree. We illustrate this in the following example.

\begin{example} \label{ex:rate-degree_tradeoff}
Continuing with Example~\ref{ex:n_15}, let $\langle \theta_1 \rangle,\dots, \langle \theta_5 \rangle$ be the minimal ideals corresponding to the conjugacy classes $C_{e},C_{y},C_{y^3},C_{y^7},C_{y^5}$, respectively. Consider the ideals $\Cs_1,\Cs_2,\Cs_3$ that provide decreasing value of degree at the cost of decreasing network coding rates

\emph{(i)} $T(\Cs_1)=\{2,3,4\}$, i.e., $\Phi(\Cs_1)=\langle \theta_2 \rangle + \langle \theta_3 \rangle + \langle \theta_4 \rangle$. This is the ideal from Example~\ref{ex:n_15} that yields network codes with rate $12/15$ and degree $6$.

\emph{(ii)} $T(\Cs_2)=\{2,3\}$. The annihilator of $\Cs_2$ is the $[15,7]$ double-error correcting BCH code with covering radius $3$; see~\cite[Table~10.1]{cohen1997covering}. Thus, network codes over $\Cs_2$ are of rate $8/15$ and degree $3$.

\emph{(iii)} $T(\Cs_3)=\{2\}$. Its annihilator is the $[15,11]$ Hamming code, which has covering radius $1$. Hence, network codes over $\Cs_3$ have rate $4/15$ and degree $1$.

A network has a solution over $\Cs_1$, $\Cs_2$, $\Cs_3$ if and only if it has a scalar linear solution over $\Fb_{2^4}$. Thus, $\Cs_1,\Cs_2,\Cs_3$ achieve a rate-complexity trade-off over the same class of solvable networks.
% (the ones that have solutions over $\Fb_{2^4}$).
\end{example}

Our next result is the observation that Hamming codes can be used as annihilators to design network codes with the smallest possible degree $\delta=1$.

\begin{lemma} \label{lem:simplex_codes}
Let $n=2^{\ell_0}-1$ for an integer $\ell_0$. If a network has a scalar linear solution over $\Fb_{2^{\ell_0}}$ then it has a rate $\ell_0/n$ circular-shift network coding solution of degree $1$.
\end{lemma}
\begin{proof}
The conjugacy class $C_y$ has exponent $\ell_0$. Let $\langle \theta_2 \rangle$ be the ideal corresponding to $C_y$, and let $\Cs$ be the ideal $\Phi^{-1}(\langle \theta_2 \rangle)$. Then $\Cs$ is the simplex code of length $n$ and its annihilator is the Hamming code. Clearly the rate of $\Cs$ is $\ell_0/n$ and the covering radius of the annihilator is $1$.
Also, \mbox{$T(\Cs)=\{2\}$} and \mbox{$q_2=2^{\ell_0}$}. Then the result follows from Lemma~\ref{lem:solution_condition_semisimple}. 
\end{proof}

\subsubsection{Comparison with Sun et al.~\cite{STLYL_TCOMM_19}}

\begin{table*}[t!]
\renewcommand{\arraystretch}{1.25}
\centering
\begin{tabular}{|l|c|c|c|c|}
\Xhline{5\arrayrulewidth}
Code & Length of Coded Packet $n$  &  Degree $\delta$ & Rate & Number of Sinks $N_{\sf rx}$ \\
\Xhline{5\arrayrulewidth}
Example~\ref{ex:rate-degree_tradeoff}, Ideal $\Cs_1$ & $15$ & $6$ & $12/15$ & $16$ \\
\hline 
Example~\ref{ex:rate-degree_tradeoff}, Ideal $\Cs_2$ & $15$ & $3$ & $8/15$ & $16$ \\
\hline
Example~\ref{ex:rate-degree_tradeoff}, Ideal $\Cs_3$ & $15$ & $1$ & $4/15$ & $16$ \\
\hline
Sun et al.~\cite[Theorem~7]{STLYL_TCOMM_19} & $15$  & $1$ & $8/15$ & $7$ \\
\hline
\hline
Lemma~\ref{lem:main_n_odd} & $7$ & $3$ & $6/7$ & $8$ \\
\hline
Lemma~\ref{lem:simplex_codes} & $7$ & $1$ & $3/7$ & $8$ \\ 
\hline
Sun et al.~\cite[Theorem~7]{STLYL_TCOMM_19} & $7$  & $1$ & $6/7$ & $3$ \\
\Xhline{5\arrayrulewidth}
\end{tabular}
\vspace{2mm}
\caption{Comparison of circular-shift network coding solutions for multicast networks.}
\label{tbl:comparison}
\end{table*}

Lemma~\ref{lem:main_n_odd} improves upon the result in~\cite[Theorem~4]{STLYL_TCOMM_19}, since the former applies to any network and the latter to only multicast networks. Our result also promises higher rate. When $n=7$, Theorem~4 of~\cite{STLYL_TCOMM_19} (see example in p.~2664) provides a rate $3/7$ network code, whereas Lemma~\ref{lem:main_n_odd} guarantees a rate $6/7$ code.

Also, Lemma~\ref{lem:main_n_odd} is similar in flavour to~\cite[Theorem~7]{STLYL_TCOMM_19}, but there are some essential differences. 
Theorem~7 of~\cite{STLYL_TCOMM_19} applies to only multicast networks and guarantees the existence of a network code where the encoding operations are of a given bounded degree $\delta$, which can be chosen by the code designer. It does not guarantee that the decoding operations are of low complexity. 
% Further, this construction requires $\ell_0 = \mathcal{O}(N_{\sf rx})$, where $N_{\sf rx}$ is the number of receivers in the multicast network. 
On the other hand, our result in Lemma~\ref{lem:main_n_odd} holds for non-multicast networks also. Our construction guarantees that the degree of the encoding as well as decoding operations is at the most $\rcov(\ann(\Cs))$. 

We now compare Lemma~\ref{lem:main_n_odd} with~\cite[Theorem~7]{STLYL_TCOMM_19} for multicast networks. 
We first recall this latter result from~\cite{STLYL_TCOMM_19}. 
Let $n$ be an odd integer and $\alpha$ be a primitive $n^\tth$ root of unity in $\Fb_{2^{\ell_0}}$. For any integer $1 \leq \delta \leq n$, let $K_{\delta}$ be equal to
\begin{align*}
 \left|\left\{ \sum_{i=0}^{n-1} a_i \alpha^i \Big| a_0,\dots,a_{n-1} \in \Fb_2, \, \wt\left(\left(a_0,\dots,a_{n-1}\right)\right) \leq \delta \right\}\right|.
\end{align*} 

That is, $K_{\delta}$ is the number of elements in $\Fb_{2^{\ell_0}}$ that can be expressed as the sum of at the most $\delta$ elements from $1,\alpha,\dots,\alpha^{n-1}$.
Theorem~7 of~\cite{STLYL_TCOMM_19} states that a multicast network has a circular-shift linear network coding solution of rate $\varphi(n)/n$ where all the encoding coefficients are of degree $\delta$ if the number of sink nodes $N_{\sf rx} \leq \frac{K_{\delta}\ell_0}{\varphi(n)} - 1$.

If we apply Lemma~\ref{lem:main_n_odd} to a multicast scenario, we require the field size $2^{\ell_0}$ to be large enough to accommodate a solution. From~\cite{JSCEEJT_IT_05} we know that $2^{\ell_0} \geq N_{\sf rx}$ is sufficient.
Hence, Lemma~\ref{lem:main_n_odd} guarantees the existence of circular-shift network coding solutions for multicast networks with up to $2^{\ell_0}$ sink nodes. 

Suppose $n$ is such that $\ell_0 \neq \varphi(n)$. We know that $\varphi(n) \geq 2 \ell_0$ (see Lemma~\ref{lem:circular_n_odd}) and $K_{\delta} \leq 2^{\ell_0}$.
Then~\cite[Theorem~7]{STLYL_TCOMM_19} can be applied to multicast networks with at the most $\frac{K_{\delta}\ell_0}{\varphi(n)} - 1 < 2^{\ell_0 - 1}$ receivers. 
Whereas Lemma~\ref{lem:main_n_odd} can be applied to multicast networks with up to $2^{\ell_0}$ receivers, which is at least two times the number of sinks that can be accommodated by~\cite{STLYL_TCOMM_19}.
For some choices of $n$ this factor could be much larger than $2$. For instance, when $n=31$, Lemma~6 serves $6$ times as many receivers as~\cite{STLYL_TCOMM_19}.
However, this ability to serve a larger number of receivers might be achieved at the cost of a larger degree in comparison to~\cite{STLYL_TCOMM_19}.
We illustrate this trade-off in Table~\ref{tbl:comparison} where we compare our new network coding solutions with those from~\cite[Theorem~7]{STLYL_TCOMM_19} for multicast networks for the cases $n=15$ and $n=7$.

\subsection{Codes from Non-Cyclic Abelian Groups} \label{sub:sec:non-cyclic-abelian}

Using non-cyclic groups can provide a wider range of choices in terms of achievable rate and degree of network codes. We illuminate this point by considering the case $n=9$ and comparing the network codes obtained from the cyclic group $G$ of order $9$ and the Abelian group $H$ which is the direct product two cyclic groups of order $3$.

The conjugacy classes of $G =\{e,y,\dots,y^8\}$, with $y^9=e$, are $C_{e}$, $C_y$ and $C_{y^3}$ which are of sizes $1$, $6$ and $2$, respectively. The three minimal ideals of $\Fb_2[G]$ have sizes $2$, $2^6$ and $2^2$, respectively. Since any ideal of $\Fb_2[G]$ is a sum of the minimal ideals, the network coding rates that are possible using this group algebra are $1/9$, $2/9$, $3/9$, $6/9$, $7/9$, $8/9$ and $1$.

Now consider the Abelian group $H$ generated by $x,y$ with $x^3=y^3=e$. The conjugacy classes of $H$ are
\begin{align*}
\{e\}, \left\{x,x^2\right\}, \left\{y,y^2\right\}, \left\{xy, x^2y^2\right\}, \left\{xy^2,x^2y\right\}.
\end{align*} 
Hence $\Fb_2[H]$ contains five minimal ideals of sizes $1,2,2,2,2$, respectively. This ring contains ideals of all possible dimensions $1,\dots,9$, and hence, provides more choice in terms of achievable rates.

We will now illustrate a rate $4/9$ network code obtained from $\Fb_2[H]$, which is not possible when using $\Fb_2[G]$.

Let $\langle \theta_1 \rangle, \dots, \langle \theta_5 \rangle$ be the ideals in $\Phi(\Cs)$ corresponding to the conjugacy classes $C_e,C_x,C_y,C_{xy},C_{xy^2}$, respectively. In the decomposition of $\Fb_2[H]$ we have $t=5$, $q_1=2$ and $q_2=\cdots=q_5=2^2$. 

Let $\Cs$ be such that $\Phi(\Cs)=\langle \theta_2 \rangle + \langle \theta_3 \rangle$. Then $\Cs \cong \{0\} \times \Fb_{2^2} \times \Fb_{2^2} \times \{0\} \times \{0\}$, $|\Cs|=2^4$. Hence, the network coding rate is $4/9$. The annihilator is 
\begin{align*}
\ann(\Cs) \cong \Fb_2 \times \{0\} \times \{0\} \times \Fb_{2^2} \times \Fb_{2^2}.
\end{align*} 
Thus, the annihilator of $\Cs$ is the set of all elements $m \in \Fb_2[H]$ such that the image $\Phi(m)=\left(\hat{m}_1,\dots,\hat{m}_5\right)$ satisfies $\hat{m}_2=\hat{m}_3=0$.

We will represent the elements of $\Cs$ as $\sum_{i=0}^{2}\sum_{j=0}^{2} m_{i,j} x^iy^j$, where \mbox{$m_{i,j} \in \Fb_2$}. If $\alpha$ is a primitive element of $\Fb_{2^2}$, then using the Fourier transform expressions~\cite{RaS_IT_92} for $\hat{m}_2$ and $\hat{m}_3$, we have
\begin{align}
0 &= \hat{m}_2 = \sum_{i=0}^{2} \sum_{j=0}^{2} \alpha^{i} m_{i,j}  \label{eq:bicyclic:1} \\
0 &= \hat{m}_3 = \sum_{i=0}^{2} \sum_{j=0}^{2} \alpha^{j} m_{i,j}. \label{eq:bicyclic:2}
\end{align} 
Now using the binary vector representation of $1,\alpha,\alpha^2 \in \Fb_{2^2}$, which are the column vectors $(1,0)$, $(0,1)$, $(1,1)$, respectively, and enumerating the components of $m$ as $\taunat(m) = (m_{0,0},m_{0,1},m_{0,2},m_{1,0},m_{1,1},m_{1,2},\dots,m_{2,2})$, we obtain the following $4 \times 9$ parity-check matrix for $\taunat(\ann(\Cs))$ from the check equations~\eqref{eq:bicyclic:1} and~\eqref{eq:bicyclic:2}
\begin{align*}
% \pmb{H} = 
\begin{bmatrix}
1 & 1 & 1 & 0 & 0 & 0 & 1 & 1 & 1 \\
0 & 0 & 0 & 1 & 1 & 1 & 1 & 1 & 1 \\
1 & 0 & 1 & 1 & 0 & 1 & 1 & 0 & 1 \\
0 & 1 & 1 & 0 & 1 & 1 & 0 & 1 & 1
\end{bmatrix}.
\end{align*} 
We know that the covering radius of $\taunat(\ann(\Cs))$ is the smallest integer $a$ such that any vector in $\Fb_2^4$ can be written as the sum of at the most $a$ columns of its parity-check matrix. A manual inspection shows that $\rcov(\ann(\Cs))=2$. We remark that $\taunat(\ann(\Cs))$ is a $[9,5]$ code, and its covering radius $2$ is the smallest among all binary linear codes of this length and dimension, see~\cite[Table~7.1]{cohen1997covering}.

To summarize, we identified a family of rate $4/9$ permute-and-add network codes of degree $2$ using an ideal $\Cs$ of $\Fb_2[H]$. A network has a solution over $\Cs$ if and only if it has a solution of $\Fb_{2^2}$.
On the other hand, if we use the cyclic group $G$, then in order to achieve a rate of at least $4/9$, we need to use an ideal of dimension $6$. The annihilator of such an ideal will be of dimension $3$, and degree of the resulting network code will be at least $3$, since the smallest covering radius among all $[9,3]$ codes is $3$.

\section{Permute \& Add Network Codes over\\ Arbitrary Finite Fields} \label{sec:arbit_finite_fields}

In Sections~\ref{sec:network_coding_over_group_codes} and~\ref{sec:existence_of_solutions} we analyzed permute-and-add network codes over the binary field $\Fb_2$, i.e., these network codes were fractional linear network codes over $\Fb_2$.
The main tools used in Sections~\ref{sec:network_coding_over_group_codes} and~\ref{sec:existence_of_solutions} are not specific to the base field over which the network code is defined.
In the current section we highlight how these techniques generalize to network codes defined over an arbitrary finite field $\Fb_q$.

Circular-shift network codes (permute-and-add network codes that use only cyclic shifts) over prime fields $\Fb_p$, $p$ a prime, were introduced in~\cite{TSWY_COMML_20} for multicast networks. 
It was shown in~\cite{TSWY_COMML_20} that a rate-$1$ circular-shift network coding solution over $\Fb_p$ exists for a multicast network if and only if the multicast network has a scalar linear solution over $\Fb_p$. Further, a sufficient condition was identified for the existence of a rate-$(n-1)/n$ circular-shift network coding solution over $\Fb_p$ for multicast networks~\cite[Proposition~8]{TSWY_COMML_20}.
In contrast to~\cite{TSWY_COMML_20}, our results in this section hold for arbitrary groups of permutations, any base field $\Fb_q$ (not necessarily prime), and any directed acyclic multigraph (not necessarily multicast networks).

\subsection{Group Algebras over $\Fb_q$}

Let $G$ be any finite group with order $n$, $q$ be any prime power and $\Fb_q$ be the finite field of size $q$. The group algebra $\Fb_q[G]$ is the set $\left\{\sum_{g \in G} a_g g~|~a_g \in \Fb_q\right\}$ with addition and multiplication defined as usual, see~\eqref{eq:addition_multiplication_group_algebra}.
%% \begin{equation*}
%% \sum_{g \in G} a_g g + \sum_{g \in G} b_g g = \sum_{g \in G} (a_g + b_g)g \text{ and } 
%% \left(\sum_{g \in G} a_g g\right) \cdot \left(\sum_{g \in G} b_g g \right) = \sum_{g \in G} (a_g + b_g)g \text{ and } 
%% \end{equation*} 
The regular representation of $g \in G$ in $\Fb_q^n$ is the $n \times n$ permutation matrix $\rhoreg_g$ over $\Fb_q$ given by~\eqref{eq:reg_group_rep}. 
The regular matrix representation of $\sum_{g \in G} r_g g \in \Fb_q[G]$ is the matrix $\sum_{g \in G} r_g \rhoreg_g$. 
This function defines an injective algebra homomorphism that maps $\Fb_q[G]$ into the matrix algebra $\Fb_q^{n \times n}$. 
Observe that 
the image of this homomorphism consists of sums of scaled permutation matrices.
As before, we will use the natural embedding $\taunat: \Fb_q[G] \to \Fb_q^n$ to represent an element $m=\sum_{g \in G} m_g g \in \Fb_q[G]$ as a length-$n$ vector $\taunat(m) = \left(m_g~|~g \in G\right)$.

\subsection{Permute-and-Add Network Codes over $\Fb_q$}

Let $\Cs$ be any ideal of $\Fb_q[G]$. Consider any $\Fb_q[G]$-linear network code over $\Cs$ with the encoding and decoding coefficients denoted by $\{k^{d,e}\}$ and $\{k^{d,i}\}$. This network code can be realized as a fractional linear network code over $\Fb_q$ as follows. 
At the source node that generates the message $Z_i \in \Cs$ all the outgoing edges carry the vector $\taunat(Z_i)$.
At every non-source node $v$ and every outgoing edge $e \in \rmout(v)$, the vector $X_e \in \Fb_q^n$ carried along $e$ is computed as 
\begin{align*}
X_e &= \sum_{d \in \rmin(v)} \sum_{g \in G} k^{d,e}_g \rhoreg_g \times X_d\\
&= \sum_{d \in \rmin(v)} \sum_{\substack{g \in G: \\ k^{d,e}_g \neq 0}} k^{d,e}_g \left(\rhoreg_g \times X_d\right),
\end{align*} 

where $k^{d,e} = \sum_{g \in G} k^{d,e}_g g$ is the expansion of $k^{d,e}$ in terms of its coefficients $k^{d,e}_g \in \Fb_q$, and $X_d \in \Fb_q^n$ is the vector carried along edge $d$.
The number of permutations and scalar multiplications applied on $X_d$ to compute $X_e$ is $\wt(k^{d,e}) = \left|\{g \in G~|~k^{d,e}_g \neq 0\}\right|$.
Finally, if $v$ is a sink node that demands the message $Z_i$, the decoding operation is realized as 
\begin{equation*}
\sum_{d \in \rmin(v)} \sum_{\substack{g \in G \\ k^{d,i}_g \neq 0}} k^{d,i}_g \left(\rhoreg_g \times X_d\right),
\end{equation*} 
where $k^{d,i} = \sum_{g \in G} k^{d,i}_g g$. If the decoding is successful, the above computation will yield $\taunat(Z_i)$.

We observe that all the encoding and decoding operations in this fractional linear network code consist of permutations and scalar multiplications applied on length-$n$ vectors. 
This is a permute-and-add network code over $\Fb_q$ with rate $k/n$ where $k$ is the dimension of $\Cs$ over $\Fb_q$. 
This network code has degree $\delta$ if $\wt(k^{d,e}), \wt(k^{d,i}) \leq \delta$ for all coding coefficients $k^{d,e}$ and $k^{d,i}$.

As in Section~\ref{sub:sec:permute-and-add}, we can use Lemma~\ref{lem:annihilator} to control the degree of the network code through the annihilator $\ann(\Cs)$. Let $\rcov(\ann(\Cs)) = \max_{\pmb{v} \in \Fb_q^n} \min_{\pmb{a} \in \taunat(\ann(\Cs))} \wt(\pmb{v} - \pmb{a})$ denote the covering radius of the subspace $\taunat(\ann(\Cs))$ in $\Fb_q^n$. 
Using the same argument that led to Theorem~\ref{thm:degree}, we deduce that every $\Fb_q[G]$-linear network code over an ideal $\Cs$ is a permute-and-add network code over $\Fb_q$ with degree $\rcov(\ann(\Cs))$.

\subsection{Existence of Permute-and-Add Network Coding Solutions over $\Fb_q$}

The highest possible network coding rate attainable using our technique is $1$, which is attained if $\Cs = \Fb_q[G]$, i.e., if the network code is a scalar linear code over $\Fb_q[G]$. Using the same idea as in the proof of Corollary~\ref{cor:existence_over_F2}, we arrive at the following generalization.

\begin{corollary}
For any finite group $G$, a network is scalar linearly solvable over $\Fb_q[G]$ if and only if it is scalar linearly solvable over $\Fb_q$.
\end{corollary}

This result generalizes~\cite[Theorem~4]{TSWY_COMML_20} (which applies to multicast networks, cyclic groups $G$, and prime fields $\Fb_p$) to arbitrary networks, finite groups and finite fields.

\subsubsection{Network Codes using Semi-Simple Abelian Group Algebras}

The structure of the algebra $\Fb_q[G]$ is completely known when $G$ is a finite group and the order $n$ of the group is relatively prime with the size $q$ of the finite field $\Fb_q$. In this case $\Fb_q[G]$ is semi-simple, and is isomorphic to a direct product of finite fields. See~\cite{RaS_IT_92} for the explicit description of the isomorphism $\dft: \Fb_q[G] \to \Fb_{q_1} \times \cdots \times \Fb_{q_t}$.
The number of finite fields in this direct product and their sizes can be determined from the conjugacy classes of $G$~\cite{RaS_IT_92}. For the group algebra $\Fb_q[G]$, the conjugacy class $C_g$ that contains the group element $g \in G$ is $C_g = \left\{g,g^q,g^{q^2},\dots,g^{q^{l-1}} \right\}$ where $l$ is the smallest integer such that $g^{q^l}=g$.
The number of finite fields $t$ in the image of the isomorphism $\dft$ is the number of distinct conjugacy classes of $G$. Suppose $C_{g_1},\dots,C_{g_t}$ are the distinct conjugacy classes, then the sizes of the finite fields are $q_k = q^{|C_{g_k}|}$, $k=1,\dots,t$.
As in Section~\ref{sub:sec:semi-simple}, there are exactly $2^t$ ideals in $\Fb_q[G]$, one ideal $\Cs$ corresponding to each choice of $T(\Cs) \subset \{1,\dots,t\}$, see~\eqref{eq:dft_isomorphism} and~\eqref{eq:R_by_annJ}.
We observe that the ideas used in the proof of Lemma~\ref{lem:solution_condition_semisimple} hold for the following generalization as well.

\begin{lemma} \label{lem:general_solution_condition_semisimple}
Let $G$ be a finite Abelian group, $q$ be a prime power that is relatively prime with $|G|$, and $\Cs$ be any ideal of $\Fb_q[G]$.
A network has a $\Fb_q[G]$-linear solution over $\Cs$ if and only if the network has a scalar linear solution over each $\Fb_{q_k}$, $k \in T(\Cs)$.
\end{lemma}

Without loss of generality let the first conjugacy class $C_{g_1}$ be the class generated by the identity element of $G$. Then, $C_{g_1} = C_e = \{e\}$, and hence, $\Fb_{q_1} = \Fb_q$.
Also, $q_k$ is a power of $q$ for each $k=2,\dots,t$, and hence, each $\Fb_{q_k}$ is an extension field of $\Fb_q$. Hence, there exists a ring homomorphism from $\Fb_q$ to $\Fb_{q_k}$, $k=2,\dots,t$.

If $\Cs$ is such that $1 \in T(\Cs)$, then a network code has a solution over $\Cs$ if and only if it has a solution over $\Fb_q$. The `only if' part follows from Lemma~\ref{lem:general_solution_condition_semisimple}. The `if' part follows from Lemma~\ref{lem:general_solution_condition_semisimple} and Theorem~\ref{thm:CoZ_1_6} using the observation that there is a ring homomorphism from $\Fb_q$ to $\Fb_{q_k}$ for all $k=2,\dots,t$.

Now suppose that $\Cs$ is such that $1 \notin T(\Cs)$. Then for any $m = \sum_{g \in G}m_g g \in \Cs$ and $\dft(m) = (\hat{m}_1,\dots,\hat{m}_t)$ we have $\hat{m}_1=0$. Since $\hat{m}_1 = \sum_{g \in G}m_g$, we conclude that for all $m \in \Cs$ we have $\sum_{g \in G}m_g = 0$. 
It is clear that for any choice of $\beta \in \Fb_q$ and any $m \in \Cs$,
\begin{align*}
\left( \sum_{g \in G} \beta g \right) \cdot \left( \sum_{g \in G} m_g g \right) &= \sum_{g \in G} \left(\sum_{h \in G} \beta m_{h^{-1}g}  \right) g \\
&= \beta \sum_{g \in G} \left(\sum_{k \in G}  m_k  \right) g\\ &= 0.
\end{align*} 

Thus, $\ann(\Cs)$ contains the set $\left\{ \sum_{g \in G}\beta g~|~\beta \in \Fb_q \right\}$, and hence, $\taunat(\ann(\Cs))$ contains the repetition code over $\Fb_q$. Therefore, $\rcov(\ann(\Cs))$ is upper bounded by the covering radius of the repetition code, which is $\lfloor n(q-1)/q\rfloor$.
We end this section with the following corollary to Lemma~\ref{lem:general_solution_condition_semisimple}.

\begin{corollary}
Let $n$ be a prime with primitive root $q$. A network has a circular-shift network coding solution over $\Fb_q$ with degree $\lfloor n(q-1)/q \rfloor$ and rate $(n-1)/n$ if it has a scalar linear solution over $\Fb_{q^{n-1}}$.
\end{corollary}
\begin{proof}
Since $n$ is a prime with primitive root $q$, the integers $1,\dots,n-1$ are powers of $q$ modulo $n$. Now let $G=\left\{e,y,y^2,\dots,y^{n-1}\right\}$ be the cyclic group of order $n$. Then $G$ has exactly two conjugacy classes, viz. $C_e=\{e\}$ and $C_y=\left\{y,y^2,\dots,y^{n-1}\right\}$. Hence, $\Fb_q[G] \cong \Fb_q \times \Fb_{q^{n-1}}$. Here $t=2$, $q_1=q$ and $q_2=q^{n-1}$. 

Now consider the ideal $\Cs$ with $T(\Cs)=\{2\}$, i.e., $\Cs = \dft^{-1}\left(\{0\} \times \Fb_{q^{n-1}}\right)$. 
% Since $1 \notin T(\Cs)$, we conclude that this network code is of degree $\lfloor n(q-1)/q \rfloor$. 
The annihilator of $\Cs$ is $\dft^{-1}\left(\Fb_q \times \{0\} \right)$, and the $1$-dimensional subspace $\taunat(\ann(\Cs))$ is the repetition code over $\Fb_q$ with covering radius $\left\lfloor n(q-1)/q \right\rfloor$.
From Lemma~\ref{lem:general_solution_condition_semisimple} a network coding solution over $\Cs$ exists if and only if a scalar linear solution over $\Fb_{q^{n-1}}$ exists. Since $\dim(\Cs)=n-1$ over $\Fb_q$, the rate of the network code is $(n-1)/n$. Finally, since $G$ is a cyclic group this is a circular-shift network code over $\Fb_q$.
\end{proof}

\section{Conclusion \& Discussion} \label{sec:conclusion}

We identified an algebraic technique to design permute-and-add network codes by using the network coding framework of Connelly and Zeger and the matrix representation of group algebras. 
The natural ring theoretic flavour of our approach allowed us to obtain new results (such as Theorem~\ref{thm:degree}, Corollary~\ref{cor:even_weight}, Lemmas~\ref{lem:solution_condition_semisimple} and~\ref{lem:simplex_codes}), and also generalize and strengthen some results known in the literature (Corollary~\ref{cor:existence_over_F2} and Lemma~\ref{lem:main_n_odd}). 
Our techniques also apply to non-cyclic Abelian groups of permutations, which yield network codes with a wider range of achievable rate and degree compared to circular-shift network codes.

The majority of our results on the existence of permute-and-add network coding solutions are for the case where the characteristic of the field does not divide the order of the group. 
This includes the case where the ideals in the group algebra correspond to BCH codes.
It might be interesting to consider the alternative. 
For instance, can we determine the existence of permute-and-add network coding solutions over $\Fb_2$ when the additive group of $\Fb_2^m$, for some integer $m$, is used as the group $G$.
In this case it is known that Reed-Muller codes $\mathcal{RM}(r,m)$ exist as ideals in the group algebra $\Fb_2[G]$, see~\cite{Ber_Cybernetics_67}; and perhaps such structural properties will be helpful in solving this problem.

In a recent work~\cite{ShH_ISIT20}, Shum and Hou designed a code for distributed storage over $\Zb_{256}$ using an ideal in the ring $\Zb_{256}/(x^5-1)$, see Example~\ref{example:Shum_Hou}. The coding operations involved in this design are byte-wise circular shifts and integer addition modulo $256$. 
Similar to our work presented in the current paper, it will be interesting to identify an algebraic approach that can generalize this construction and relate the existence of such coding solutions to the existence of appropriate scalar linear solutions.

\appendices

\section{Proof of Lemma 1} \label{app:lem:annihilator}
\begin{proof}
We will use induction on the topologically ordered list of vertices in the network to show that if the original network code is replaced with the modified code, then the value of the symbols carried on the network edges do not change. 
This statement is clearly true for the outgoing edges of all the source nodes, since these symbols are not coded.

For each edge $e$, let the value of the symbol carried on this edge under the original and modified codes be $X_e$ and $X_e'$, respectively.
To complete the induction, consider a non-source node $v$. For each $d \in \rmin(v)$, assume $X_d=X'_d$. Then, for each $e \in \rmout(v)$, we have
\begin{align*}
% % {\textstyle 
\tau^{-1}\left(X_e'\right) &=  \sum_{d \in \rmin(v)} \left(k^{d,e} + a^{d,e}\right) \, \tau^{-1}(X_d) \\
% &=  \sum_{d \in \rmin(v)} (k^{d,e} + a^{d,e}) \, \tau^{-1}(X_d) \\
&=  \sum_{d \in \rmin(v)} k^{d,e}\tau^{-1}(X_d) + a^{d,e}\tau^{-1}(X_d) \\
&= \sum_{d \in \rmin(v)} k^{d,e}\tau^{-1}(X_d) = \tau^{-1}(X_e),
\end{align*} 
and hence, $X_e' = X_e$. 

Finally, we observe that the decoding function at a node $v$ that demands a message $Z_i$ in the modified code is
\begin{align*}
\sum_{d \in \rmin(v)} &\left(k^{d,i} + a^{d,i}\right) \tau^{-1}(X_d)\\
&=\sum_{d \in \rmin(v)} k^{d,i} \, \tau^{-1}(X_d) + a^{d,i}\tau^{-1}(X_d) \\
&=\sum_{d \in \rmin(v)} k^{d,i}\, \tau^{-1}(X_d)\\ 
&= Z_i.
\end{align*} 
Hence, the sinks can decode their demands in the modified network code.
\end{proof}

\section{Degree of the Network Code in Example~\ref{ex:n_15}} \label{app:ex:n_15}

The annihilator of this code is $\ann(\Cs) = \Fb_2 \times \{0\} \times \{0\} \times \{0\} \times \{\Fb_4\}$. This is a cyclic code whose parity-check polynomial $h(x)$ is the product of the minimal polynomials of $\omega^0$ and $\omega^5$, where $\omega$ is the primitive $5^\tth$ root of unity~\cite{macwilliams1977theory}. Thus 
\begin{equation*}
h(x) = \left(x-1\right) \left(x-\omega^5\right) \left(x-\omega^{10}\right) = \left(x+1\right)\left(x^2+x+1\right).
\end{equation*} 
\enlargethispage{-2.4in}
Hence, the generator polynomial of $\ann(\Cs)$ is $g(x) = \left(x^{15}+1\right)/h(x) = x^{12}+x^{9}+x^6 + x^3 + 1$. This code has generator matrix
\begin{equation*}
 \begin{bmatrix} \pmb{I}_3 & \pmb{I}_3 & \pmb{I}_3 & \pmb{I}_3 & \pmb{I}_3 \end{bmatrix}.
\end{equation*} 
Up to coordinate permutations this code is equivalent to a direct sum of three length-$5$ repetition codes. Applying~\cite[Theorem~3.2.1]{cohen1997covering} we deduce that the covering radius of $\ann(\Cs)$ is equal to $3$ times the covering radius of the length-$5$ repetition code. Thus, the degree of the network code $\Cs$ is $\rcov(\ann(\Cs))=6$.
% % % % %
%%%% references %%%%% 
\bibliographystyle{IEEEtran}
%\bibliography{IEEEabrv,Permute_and_Add_TCOM_21}

%% Biographies
\vspace*{-16\baselineskip}
\begin{IEEEbiographynophoto}{Lakshmi Prasad  Natarajan}
is an Assistant Professor in the Department of Electrical Engineering, Indian Institute of Technology Hyderabad. He received the B.E.\ degree from the College of Engineering, Guindy, in electronics and communication in 2008, and the Ph.D. degree from the Indian Institute of Science, Bangalore, in 2013. Between 2014 and 2016 he held a post-doctoral position at the Department of Electrical and Computer Systems Engineering, Monash University, Australia. His primary research interests are coding techniques and information theory for communication systems.
\end{IEEEbiographynophoto}
\vspace*{-16\baselineskip}
\begin{IEEEbiographynophoto}{Smiju Kodamthuruthil Joy}
 received the AMIE degree in electronics and communication from Institution of Engineers, Kolkata, in 2007 and the M. Tech. degree from National Institute of Technology, Rourkela, in 2010. Currently he is a Ph.D. student in the Department of Electrical Engineering, Indian Institute of Technology Hyderabad. His research interests include index coding and network coding.
 \end{IEEEbiographynophoto}
\end{document}